\documentclass[a4paper,12pt,reqno]{amsart}

\usepackage[utf8]{inputenc}
\usepackage{amsmath,amsfonts,amssymb,amsopn,amscd}
\usepackage{comment}
\usepackage{dsfont}
\usepackage{graphicx}
\usepackage{color}
\usepackage[colorlinks]{hyperref}
\usepackage{amsthm}
\usepackage{epigraph}
\usepackage{todonotes}
\usepackage[left=3.5cm,top=2.5cm,bottom=2cm,right=3cm]{geometry}
\usepackage{tikz-cd}

\usepackage[normalem]{ulem}

\allowdisplaybreaks[4]

\DeclareRobustCommand{\SkipTocEntry}[5]{}

\DeclareMathOperator{\eqlaw}{\stackrel{\Lc}{=}}

\DeclareMathOperator{\Cov}{Cov}
\DeclareMathOperator*{\argmax}{arg\,max}

\DeclareMathOperator{\id}{id}

\def\1{{\mathbf 1}}

\def\pa{\partial}

\def\m{\mu}

\def\vare{\varepsilon}

\def\ve{\varepsilon}

\def\Q{{\mathbb Q}}
\def\R{{\mathbb R}}

\def\P{{\mathbb P}}
\def\E{{\mathbb E}}

\def\Ac{{\mathcal A}}

\def\Fc{{\mathcal F}}

\def\Lc{{\mathcal L}}

\def\Nc{{\mathcal N}}

\def\Sc{{\mathcal S}}

\def\Wc{{\mathcal W}}

\def\ab{{\mathbf a}}
\def\ub{{\mathbf u}}
\def\rb{{\mathbf r}}
\def\vb{{\mathbf v}}
\def\wb{{\mathbf w}}

\def\fP{\mathbf{P}}

\newtheorem{thm}{Theorem}[section]
\newtheorem{proposition}[thm]{Proposition}

\newtheorem{definition}[thm]{Definition}

\newtheorem{lemma}[thm]{Lemma}

\newtheorem{rmk}[thm]{Remark}

\numberwithin{equation}{section}
\numberwithin{figure}{section}

\begin{document}
\title[Gaussian Kyle Model with  imperfect information and risk aversion ]
      { Solvability of the Gaussian Kyle model with  imperfect information and risk aversion }
      
\author{Reda Chhaibi}

\author{Ibrahim Ekren}
\thanks{I. Ekren is partially supported by the NSF grant DMS-2406240}

\author{Eunjung Noh}

\begin{abstract}
We investigate a Kyle model under Gaussian assumptions where a risk-averse informed trader has imperfect information on the fundamental price of an asset. 
We show that an equilibrium can be constructed by considering an optimal transport problem that is solved under a measure that renders the utility of the informed trader martingale and a filtering problem under the historical measure.
\end{abstract}


\maketitle


\hrule
\tableofcontents
\hrule


\section{Introduction}
The relations between the concepts of liquidity, price impact, information asymmetry, and adverse selection are some of the main focuses of market microstructure research. The seminal paper by Kyle \cite{kyle} provided a simple, tractable framework and has become a corner stone of this research area. In this model, there are three types of trading groups, namely an informed trader, noise traders, and a market maker. A risky asset is traded over a single period, where the informed trader has private information from the beginning on the normally-distributed terminal price. Based on this information, she determines the size of her market order. The noise traders, on the other hand, place orders without any information on the price. The market maker observes the sum of orders, called order flow, without the possibility of distinguishing them. Based on this observation, the market maker decides on the price of the asset, and clears the orders. In this framework, Kyle showed the existence of an equilibrium, in which the informed trader’s strategy is linear with respect to the fundamental price, and the market maker’s pricing rule is linear with respect to the order flow. 

Numerous extensions have been studied, and we briefly survey just a few of them. 
A continuous-time version of the Kyle model with non-Gaussian assumption on the terminal price was first proposed by Back \cite{back1992}. 
Many other extensions of the problem with risk-neutral informed trader has been studied in the literature \cite{bp,las2007,cc2007,ccetin2023insider,corcuera2010kyle,bia,cx,cd1,cdf,bal,c,ckl23,cc24}.

The informed trader being a single agent, it is crucial to allow the risk aversion.   
\cite{hs1994,Subrahmanyam98,baruch,cho,bose2020kyle,bose2021multidimensional,RyanDonnelly,xs24} study a version of Kyle's model with a risk-averse informed trader. 
However, all these papers make a simplifying assumption that seems unreasonable. They all assume that the informed trader has perfect information on the final price. In the case of a risk-neutral insider, this assumption is not consequential. Indeed, assume that the informed trader is risk-neutral and has a partial information on the terminal price. Then, the informed trader will treat the conditional expectation of the terminal price given her private information as ``the terminal price" and an equilibrium exists. In this equilibrium, at the terminal time, there is a jump in the price process that corresponds to the deviation between the true terminal price and the conditional expectation of the terminal price given the informed trader's private information. This jump, which is of $0$ mean for all agents, does not affect the optimization problem of all agents. 

If the informed trader is risk-averse, this argument does not work. Indeed, at the final time the informed trader knows that there will be a jump in the price. Thus, the position she builds increases the risk that this jump generates. Thus, in the presence of terminal price risk, a risk-averse informed trader has to make a trade-off between the risk generated by the position that she builds and the returns from this position. 

The mathematical difficulty arises from the fact that one cannot reduce this problem to the control of a one-dimensional diffusion process as in \cite{kyle,back1992,back2020optimal}, where it is shown that different ingredients in the equilibrium are functions of total demand. It is also shown in \cite{cho,bose2020kyle,bose2021multidimensional} that a transformation of the total demand reduces the problem to the control of a novel state variable.

In this paper, we show that an equilibrium exists in the Kyle model with a risk-averse informed trader that has partial information on the terminal price. We make the assumption that the terminal price is Gaussian and show that there is an equilibrium in which the informed trader controls both an endogenously determined process introduced in \cite{cho,bose2020kyle,bose2021multidimensional} and the conditional distribution of the price given the information of the market maker. 
In our work, the rational pricing condition is similar to \cite{bal}. This means that the terminal price quoted by the market maker is equivalent to the utility indifference price of the informed trader. Given this definition of the terminal price, we compute the expected profit of the market maker, which turns out to be negative. Thus, the market maker has to be compensated for fulfilling this role. 

We show that when the informed trader is more informed about the terminal price (meaning the variance of difference between the terminal price and the observed price decreases to $0$), the equilibrium converges to its counterpart in \cite{cho,bose2020kyle} and both Kyle's Lambda and price volatility increase to their values in \cite{cho,bose2020kyle}. 
Thus, a higher uncertainty of the informed trader in the terminal price makes the market less volatile and the market maker less reactive to changes in demand. 

Our methodology for finding the equilibrium is based on \cite{first_paper,back2020optimal}.
In \cite{back2020optimal}, it was shown that several important quantities in the Kyle-Back model can be described using the theory of optimal transport. The connection between optimal transport theory and the Kyle-Back model relies on the \textit{inconspicuous trading property} of the equilibrium, which means that the informed trader's strategy remains undetected to the market maker. Such a trading strategy imposes a distributional constraint on the total order flow $Y$ that the market maker receives at the terminal time, i.e, $Y$ is a Brownian motion in its filtration and $Y_T$ has a Gaussian distribution. In \cite{back2020optimal}, in equilibrium the dependence between the terminal price and $Y_T$ can be characterized by an optimal transport problem between this Gaussian distribution of $Y_T$ and the distribution of the price.

In \cite{first_paper}, we considered a financial market with multi-assets having general distributions and fixed trading horizon, and provided a criterion to establish the existence of equilibrium when the criterion of the informed trader depends non-linearly on the private information and the total demand. Due to the risk-aversion and partial information of the informed trader, this nonlinear dependence also appears in this follow-up paper; see the decomposition of utility in \eqref{eq:exputi1}. However, unlike \cite{first_paper} where the relevant state process to construct the equilibrium is the total demand, in this paper we construct this state process via the methodology of \cite{cho,bose2020kyle,bose2021multidimensional} and a fixed point condition involving an optimal transport problem and a change of measure. With the assumption of Gaussianity, we unravel a solvable structure so that if the variance of $Y_T$ is the root of a real function, an equilibrium exists.

The rest of the paper is organized as follows. In Section \ref{s.MainResult}, we state the main result of this paper, which is the explicit expression of equilibrium and its property in our Gaussian Kyle model with imperfect information. To that end, we preempt multiple aspects of the paper by directly introducing the relevant quantities and equations.
In Section \ref{s.PricingRule}, we analyze the equilibrium by studying the value function of the problem and the HJB equation it satisfies. There, we shall see that thanks to a new absolutely continuous measure $\Q$, the dynamic optimization problem can be recast into a terminal problem. This measure $\Q$ is a novel feature of the risk-averse case.
In Section \ref{sect:StructureOptimizer}, we invoke the theory of optimal transport in order to characterize the point-wise optimizer of the terminal problem -- under $\Q$. 
Also, in the current Gaussian setup, we explicitly compute Kantorovich potentials and optimal transport map.
While the solvability of Gaussian optimal transport is no surprise, the price structure is explicitly computed as well.
Most proofs of the results are given in Section \ref{s.appendix}.


\subsection{Notations} Throughout this paper we will use the following notations. 
For a vector $\ub$, the transpose is written as $\ub^\top$.
We denote by $e_1$ the vector  $\begin{pmatrix}
1, 0
\end{pmatrix}^T$.
     For two probability measures $\P$ and $\Q$ and for any stochastic process $X$, $\mathcal{L}(X)$ (and $\mathcal{L}^\Q (X)$) means the law of $X$ under $\P$ (and $\Q$, resp.).
    Similarly, we do not write $\P$ for expectation and variance under $\P$ (e.g. $\mathbb{E}[X] = \mathbb{E}^{\P}[X]$).

\section{Model}
Consider a finite time horizon $[0,T]$.
Let $(\Omega, \Ac,\P)$ be a probability space endowed with a one dimensional Brownian motion $B$ and one dimensional Gaussian random variables $ \Xi ,\beta, G$ with means $m_\Xi ,m_\beta,0$ and variances $\sigma_\Xi^2,\sigma_\beta^2,\varepsilon>0$. We assume that $B, \Xi ,G,\beta$ are independent and define $\Fc:=(\Fc_t)_{t\in [0,T]}$ as the filtration generated by $(B, \Xi ,\beta)$ and augmented with null sets in the sense of \cite[Definition 7.2 of Chapter 2]{karatzas2012brownian}. The random variable $G$ is $\Ac$-measurable and independent from $\Fc_T\subset \Ac$.

We consider a financial market where the interest rate is zero and a single asset is continuously traded with fixed maturity $T>0$. We suppose that the fundamental value of the asset will be announced to be $V_f= \Xi -G$ at time $T$ and $G$ is interpreted as an observation noise and $ \Xi $ as the noisy observation value of the fundamental value $V_f$ of the asset.  

There are three groups of market participants, namely an informed trader, noise traders, and a market maker. 
The noise traders, whom we group under the name `\textit{noise trader}', trade due to exogenous needs and their position at time $t \in [0, T]$ is assumed to be $Z_t = \sigma B_t$ for some fixed $\sigma >0$.  
The risk-averse {\it informed trader} learns the value of $ \Xi $ at time $t=0$, whereas $G$ is inaccessible to all market participants. The informed trader has $\beta$ as a random initial endowment in the asset. 
The cumulative demand of the informed trader is denoted by $X_t$ for $t\in[0,T]$ with $X_0=0$ and the informed trader's position in the asset is $X_t+\beta$. The risk-neutral {\it market maker} only observes the total demand
$Y_t = X_t+Z_t \ $
and quotes a price for the asset $P_t$. The information of the market maker is the public filtration generated by $Y$ and is denoted by $\Fc^Y_t$.  
At time $t$, the informed trader has access to her initial endowment $\beta$, observed value $ \Xi $, and the path of the price process $(P_s)_{s\in [0,t]}$. In all equilibria considered in the paper, the informed trader can compute the value of $Y_t$ from the path of $P$, and hence can obtain $Z_t$. Thus, the information of the informed trader at time $t$ is $\Fc_t $ with $\Fc^Y_t\subset \Fc_t$.

We denote the law of $ \Xi $ and $Z_T-\beta$ as  $\nu$ and $\mu$, respectively. At the initial time, the market maker only knows these laws, whereas the informed trader knows the value of $( \Xi ,\beta)$.
Denoting the constant absolute risk aversion (CARA) parameter of the informed trader as $\gamma>0$, we also define $g$ by 
$$g(x)=\frac{1}{\gamma}\ln\E[e^{\gamma Gx}]=\frac{\gamma \ve x^2}{2} $$
as such $\vare$ is the variance of the Gaussian centered random variable $G$. 


\subsection{Statement of the problem}
The main goal of this paper is to show the existence of an equilibrium between the market participants using the optimal transport theory and a filtering. We first define the optimization criterion and the set of admissible strategies of each participant. 

We denote by $\Sc(\Fc)$, the set of continuous $\Fc$-semimartingales, and define the set of pricing rules of the market maker as follows. 

\begin{definition}
    A pricing rule is a mapping $\mathbf{P}:\Sc(\Fc)\mapsto \Sc(\Fc)$ so that for all $S\in \Sc(\Fc)$, $\mathbf{P}(S)$ is semimartingale with respect to the filtration generated by $S$ and $\mathbf{P}$ is non-anticipative in the sense that for all $S,S'\in \Sc(\Fc)$ with $S_s=S'_s$ for $s\in [0,t]$ a.s., we have $\mathbf{P}_t(S)=\mathbf{P}_t(S')$.
\end{definition}

In the equilibrium we construct, the equilibrium price is obtained by evaluating a pricing rule $\mathbf{P}$ to the total demand $Y$ as
$$P_t:=\mathbf{P}_t(Y) \ . $$
We now define the set of admissible strategies of the informed trader given a pricing rule $\mathbf{P}.$
\begin{definition}
    Given a pricing rule $\mathbf{P},$ an admissible trading strategy for the informed trader is an $\Fc$-adapted continuous semimartingale $X$ so that $X_0=0$ and 
    $$\E\left[ e^{\frac{1}{2} \int_0^T \gamma^2 \sigma^2 \mathbf{P}_t^2(X+Z)dt} \right]<\infty \ .$$
    We denote by $\Ac(\mathbf{P})$ the set of admissible trading strategies given $\mathbf{P}.$
\end{definition}

As mentioned, the strategy of the informed trader satisfies $X_0=0$ so that $X_t$ represents the number of shares purchased by the informed trader up to time $t$ and $X_t+\beta$ is her total position in the asset. As in \cite{back1992}, the informed trader's wealth at time $T$ for a given strategy $X$ is
\begin{align*}
\Wc_T(\beta,X,\fP) :&=  \ (X_T+\beta)(V_f-P_T) +\int_0^T   (X_s+\beta) dP_s+\beta P_0 \ , 
\end{align*}
where $(X_T+\beta)(V_f-P_T) $ is the profit of the informed trader from the price mismatch at time $T$ between quoted price $P_T$ and the revealed fundamental value $V_f$, the second term $\int_0^T   (X_s+\beta) dP_s$ is her profit from trading, and $\beta P_0$ is the marked to market value of her initial position. 
Thus, by integration by parts, 
\begin{align*}
\Wc_T(\beta,X,\fP) &= (X_T+\beta)V_f - \int_0^T  \fP_s(X+Z) dX_s- \langle X,\fP(X+Z)\rangle_T \notag\\
&= (X_T+\beta) \Xi -(X_T+\beta) G  - \int_0^T  \fP_s(X+Z) dX_s- \langle X,\fP(X+Z)\rangle_T\ , \label{eq:wealth}
\end{align*}
where we omitted the dependence on $V_f$ for the expression of the profit. 
Thus, the goal of the informed trader is to optimize
\begin{equation}
    \label{eq:problem}
    \sup_{X\in \Ac(\mathbf{P})}\E\left[- \exp\left( -\gamma \Wc_T(\beta,X,\fP) \right)\ | \ \Fc_0 \right] \ , 
\end{equation}
where the supremum is among the class of admissible trading strategies given $\fP$.
If we add $1$ and divide \eqref{eq:problem} by $\gamma$, we obtain that as $\gamma\to 0$ the problem in \eqref{eq:problem} converges to the risk neutral problem 
\begin{equation*}
    \label{eq:problem0}
    \sup_{X\in \Ac(\mathbf{P})}\E\left[ \Wc_T(\beta,X,\fP)\ | \ \Fc_0 \right]
\end{equation*}
studied in \cite{kyle,back1992,first_paper}. All our results also hold for $\gamma=0$, and we obtain the same pricing rule and the strategy as in \cite{first_paper}. 

\begin{definition}
A pair $(X^*,\mathbf{P}^*)$ of  pricing rule $\mathbf{P}^*$ and admissible trading strategy $X^*\in \Ac(\mathbf{P}^*)$ is an equilibrium if the following conditions are satisfied. 

\begin{enumerate}
    \item If the market maker uses the price $P_t=\fP^{*}_t(X+Z)$, then $X^*$ is a maximizer of
\begin{align}\label{eq:optimality}
\sup_{X\in \Ac(\mathbf{P}^*)}\E\left[ - \exp\left( -\gamma \Wc_T(\beta,X,\fP^*) \right)\ | \ \Fc_0 \right] \ .
\end{align}

    \item If the informed trader uses the strategy $X^*$, then the price satisfies the rationality condition
\begin{align}\label{eq:rationality}
    \fP^*_t({X^*+Z})=\E[V_f-\gamma\vare (X_T^*+\beta)|\Fc^{X^*+Z}_t]=\E[ \Xi -\gamma\vare (X_T^*+\beta)|\Fc^{X^*+Z}_t] \ .
\end{align}
\end{enumerate}

\end{definition}

Note that our rationality condition is different from the usual definition $$P_t=\E[V_f|\Fc^{X^*+Z}_t]=\E[ \Xi |\Fc^{X^*+Z}_t]$$ in the literature such as \cite{kyle,back1992}. This classical definition can be justified by the fact that there are many competing market makers and the competition between these agents imposes a $0$ expected profit condition for the market makers. 
The position of a representative market maker is $-Y_t$ and his realized profit, coming from his liquidation at terminal time and trading on $[0,T]$, is 
$$ 
    -Y_T (V_f-P_T)-\int_0^T Y_s dP_s \ . 
$$
As such, the classical martingality and terminal value condition $P_t=\E[V_f|\Fc^Y_t]$ indeed leads to a $0$ conditional expected future profit for the market maker. 

In our context, we use the definition in \cite{bal,first_paper}. Similar to \cite{bal}, our definition of rationality \eqref{eq:rationality} can in fact be deduced from \eqref{eq:optimality}.
Indeed, we assume that $\fP^*$ is the equilibrium pricing rule and $X^*$ is an equilibrium trading strategy. Thanks to the independence of $G$ and an application of integration by parts, for any admissible strategy $X$, the expected utility of the informed trader admits the expression 
\begin{align}
&\E\left[ - \exp\left( \gamma (X_T+\beta) G -\gamma (X_T+\beta)( \Xi -P_T) -\gamma \int_0^T  X_s dP_s  - \gamma \beta P_T\right) | \ \Fc_0 \right]=\notag\\
&\E\left[ - \exp\left(\frac{\gamma^2 \varepsilon (X_T+\beta)^2}{2} -\gamma (X_T+\beta)( \Xi -P_T) -\gamma \int_0^T  X_s dP_s - \gamma \beta P_T \right) | \ \Fc_0 \right]\label{eq:exputi1}
\end{align}
where $P=\fP(X+Z)$. 
In the equilibrium we construct, up to division by $\gamma$, the term in the exponential admits the expression 
$$\gamma \varepsilon \frac{(X_T+\beta)^2}{2}- \Xi  (X_T+\beta)+u(T,\chi_T) - u(0,0) -\int_0^T  P_s dZ_s - \int_0^T \frac{\gamma \sigma^2 P_s^2}{2}ds$$
for some function $u(T,\cdot)$ and a process $(\chi_t)$ driven by $Y$ which will be specified in Section \ref{s.PricingRule}. Note that this expression is $\Fc_T$-measurable. In the equilibrium we construct, we can consider perturbations of the optimum $X^*$ that are equal to $X^*$ on $[0,T-\epsilon]$ and change the value of $X^*_T$ by trading differently on $[T-\epsilon,T]$. Given the equality $\fP_T(X+Z)=\pa_\chi u(T,\chi_T)$ from Section \ref{s.PricingRule}, informally, the first-order optimality condition in these perturbation directions can be written by setting the quantity below to $0$
\begin{align*}
&\frac{\pa}{\pa X_T}\left( e^{\gamma \varepsilon \frac{(X_T+\beta)^2}{2}- \Xi  (X_T+\beta)+u(T,\chi_T) - u(0,0) -\int_0^T  P_s dZ_s - \int_0^T \frac{\gamma \sigma^2 P_s^2}{2}ds}\right)\biggr\rvert_{X=X^*}\\
&=\left(\gamma \varepsilon (X^*_T+\beta)- \Xi +\fP_T(X^*+Z)\right)\\
& \quad \times e^{\gamma \varepsilon \frac{(X_T^*+\beta)^2}{2}- \Xi  (X_T^*+\beta)+u(T,\chi_T) - u(0,0)
-\int_0^T  P_s dZ_s - \int_0^T \frac{\gamma \sigma^2 P_s^2}{2}ds} \ .
\end{align*}
Note that all the above quantities are known by the informed trader at time $T-$. Thus, if the equality $ \Xi -\gamma \varepsilon (X^*_T+\beta)=\fP_T(X^*+Z)$ does not hold, the informed trader will deviate from $X_T^*$. This heuristic shows that without this rationality condition one cannot have an equilibrium. Surprisingly, the equilibrium condition $\gamma \varepsilon (X^*_T+\beta)- \Xi +\fP_T(X^*+Z)=0$ also brings mathematical convenience. Indeed, this definition of equilibrium is also what we need to use the optimal transport methodology of \cite{back2020optimal,first_paper}.

\section{Main result}\label{s.MainResult}
In this section, we state our main existence of equilibrium result and discuss the properties of the equilibrium. Recall that
$$\mathcal{L}( \Xi ) = \mathcal{N}(m_\Xi , \sigma_\Xi^2),\, \quad
\mathcal{L}(Z_T - \beta) = \mathcal{N}(-m_\beta, \sigma^2 T + \sigma_\beta^2),\, \quad
\mathcal{L}(G)= \mathcal{N}(0,\varepsilon)$$
and denote
$$
\sigma_{e} = \sqrt{\sigma^2_\Xi+\vare^2\gamma^2( \sigma^2 T+\sigma^2_\beta)} \ , \quad 
\ub :=
\begin{bmatrix}
\varepsilon\gamma(\sigma^2T+\sigma_\beta^2) \\ \sigma^2_\Xi
\end{bmatrix} \ , \quad
\wb :=
\begin{bmatrix}
\varepsilon\gamma \\ 1
\end{bmatrix},\quad \vb:=
    \begin{bmatrix}
    1 \\ - \varepsilon\gamma
    \end{bmatrix} \ .
$$
\subsection{Construction of an equilibrium}
In order to construct an equilibrium, we need the following two Lemmata whose proofs are provided in the Section \ref{proof-1}-\ref{proof-2}. 


\begin{lemma}\label{lem:v}
    All other parameters being positive, for all $\vare,\gamma\geq 0$ there exists a unique positive solution $v = v\left( \varepsilon, \gamma, \sigma^2 T, \sigma_\Xi^2, \sigma_\beta^2 \right) > 0$ to the equation
    $$ 
   0=
    \frac{\left( 
     \frac{\sigma_{e}}{v}-\vare \gamma
     \right)^2}
     {1 - \gamma\sigma^2 T (\frac{\sigma_{e}}{v}-\vare\gamma) } 
     - 
     \frac{2\vare }{\sigma^2 T } 
     \log \left( 1-\gamma\sigma^2 T \left(\frac{\sigma_{e}}{v}-\vare\gamma\right)  \right)
     -
      \left( \frac{\sigma_\Xi^2}{\sigma^2 T}+\vare^2\gamma^2 \frac{\sigma_\beta^2}{\sigma^2 T} \right) 
$$
satisfying 
\begin{equation}\label{bound:v}
1>\frac{v\vare\gamma}{\sigma_{e}}>\frac{1}{1+\frac{1}{\vare\gamma^2 \sigma^2 T}} \ .
\end{equation}
\end{lemma}

\begin{rmk}
    (i) Thanks to \eqref{bound:v}, we define 
    \begin{align}\label{eq:terminallambda} 
    \Lambda:=\frac{\sigma_{e}}{v}-\vare\gamma>0
    \end{align}
    which will turn out to be the value of Kyle's lambda, i.e. sensitivity of the price in the total demand $Y$, at time $t=T$, see Eq. \eqref{eq:pricedyn}.

    (ii) For $\vare=0$, we have that $ \Xi =V_f$ and $v$ solves the quadratic equation
        $$ \frac{v^2}{\sigma^2 T} 
   =
    \frac{1}
     {1 -  \frac{\sigma_\Xi\gamma\sigma^2 T}{v} } \ .
$$
In this case, the solution of this equation is $v= \frac{1}{2}\sigma_\Xi \gamma \sigma^2 T + \frac{\sigma\sqrt{T}}{2}\sqrt{4+\gamma^2 \sigma_\Xi^2 \sigma^2 T}$, which is in agreement with \cite[Proposition 3]{cho}.

(iii) For $\gamma=0$, the informed trader is risk neutral and optimizes her expected wealth and $v^2=\sigma^2 T$ leads to the classical Kyle model with Gaussian information \cite{kyle}.
\end{rmk}
In order to construct an equilibrium, given $v$ in Lemma \ref{lem:v} and $\Lambda$ in \eqref{eq:terminallambda}, define the following deterministic continuous functions
\begin{align}\label{eq:defdet}
    p_t & = \frac{1}{\Lambda^{-1}-{\gamma\sigma^2}(T-t)} \ ,\quad 
    k_t    = \frac{1-\frac{v\vare\gamma}{\sigma_e}\gamma\sigma^2\Lambda(T-t)}{1-\gamma\sigma^2\Lambda(T-t)} \ ,\quad 
    \rb_t  = \frac{k_t\ub}{\sigma_e v} \\ \, \notag
    \Sigma_t &=\begin{bmatrix}
     \sigma^2 t + \sigma_\beta^2 & 0\\ 
    0           & \sigma^2_\Xi
   \end{bmatrix}
   -
   \frac{\sigma^2 \int_0^t k_{s}^2 ds}
        {\sigma_e^2 v^2}
   \ub \ub^\top. 
   \end{align}
   Define also $m$ as the root of the equation 
   $$-m+\frac{m_\Xi -\vare\gamma m_\beta - \frac{\sigma_e}{v}m}{\Lambda} (k_0-1)=0$$ and
   \begin{align*}
 q_t   &= \frac{m_\Xi -\vare\gamma m_\beta - \frac{\sigma_e}{v}m}{1-\Lambda{\gamma\sigma^2}(T-t)} \ , \quad  
        \ab_t  = \begin{pmatrix}-m_\beta\\m_\Xi \end{pmatrix} + \frac{b_t}{\sigma_e v} \ub \ ,\\
     b_t  &=  -m + \int_t^T\gamma\sigma^2 q_s \left( k_s - \frac{v \varepsilon \gamma}{\sigma_e} \right) ds \ . \nonumber
   \end{align*}
Additionally, define the projection
$$x\in \R^2\mapsto\Phi(x)   = x - \frac{\wb^\top x}{\wb^\top \ub } \ub =x - \frac{\wb^\top x}{\sigma_e^2 } \ub$$
and 
$ \left( \Phi_{s,t} \ ; \ s \leq t < T \right)$ the flow of the equation
\begin{align}\label{eq:resolvent}
    dx_t=-\sigma^2 \rb_t  \left(\rb_t -e_1\right)^\top \Sigma_t^{-1}x_t dt
\end{align}
so that $t\mapsto\Phi_{s,t}(x)$ solves this equation with initial value $\Phi_{s,s}(x)=x\in \R^2$. 

The following Lemma lists properties of these quantities that we will use. As we will see, $\Sigma_T$ will not be invertible. In particular, this non-invertibility of $\Sigma_T$ requires us to explicitly compute the expansion of $\Sigma^{-1}_t$ as $t\to T$. This point will be needed to study the mean-reversion properties of $\Phi_{s,t}$, which will play an important role.

   \begin{lemma}\label{lem:sigma}
   Given the choice of $v$ as in Lemma \ref{lem:v},
 $\Sigma_t$ is invertible for $t<T$ with inverse given by
\begin{align*}
    \Sigma_t^{-1}  & =     \begin{bmatrix}
     \sigma^2 t + \sigma_\beta^2 & 0 \\
    0           & \sigma^2_\Xi
   \end{bmatrix}^{-1} \\ \notag
  & \quad  +
   \frac{\sigma^2 \int_0^t k_{s}^2 ds}
        {\sigma_{e}^2 v^2}
   \frac{
   \begin{bmatrix}
     \sigma^2 t + \sigma_\beta^2 & 0 \\
    0           & \sigma^2_\Xi
   \end{bmatrix}^{-1}
   \ub
   \ub^\top
   \begin{bmatrix}
     \sigma^2 t + \sigma_\beta^2 & 0 \\
    0           & \sigma^2_\Xi
   \end{bmatrix}^{-1}
   }
   {
    1 - f_{t}
    } \ , \\ 
    f_{t} &:= 
   \frac{\sigma^2 \int_0^t k_{s}^2 ds}
        {\sigma_{e}^2 v^2}
   \ub^\top
   \begin{bmatrix}
     \sigma^2 t +  \sigma_\beta^2 & 0 \\
    0           & \sigma^2_\Xi
   \end{bmatrix}^{-1}
   \ub \ .\notag
\end{align*}
In particular,  $\Sigma_T=\frac{(\sigma^2 T +\sigma_\beta^2)\sigma^2_\Xi}{\sigma_{e}^2 }
     \vb\vb^\top$ and we have the asymptotics
\begin{align}\label{eq:sigma-equivalent}
\notag
\wb^\top\Sigma_t^{-1}\wb& \sim_{t\uparrow T}\frac{
   |\wb|^4
   }
   {  \sigma^2(T-t)\left(\frac{\sigma_{e}^2}{v^2}-\vare^2\gamma^2\right)} \ ,\quad
   \vb^\top\Sigma_t^{-1}\vb \sim_{t\uparrow T} \frac{
   \sigma_{e}^2
   }
   {  (\sigma^2 T +\sigma_\beta^2)\sigma^2_\Xi} \ ,\\ 
   \wb^\top\Sigma_t^{-1}\vb& \sim_{t\uparrow T}\gamma\vare \left(\frac{1}{\sigma^2 T+\sigma^2_\beta}-\frac{1}{\sigma^2_\Xi}\right)+\frac{2(\gamma\vare)^3
   }
   { (\sigma^2 T+\sigma^2_\beta)\left(\frac{\sigma_{e}^2}{v^2}-\vare^2\gamma^2\right)} \ .
\end{align}
Additionally, $ \left( \Phi_{s,t} \ ; \ s \leq t < T \right)$ is well defined and satisfies 
for all $(s,x)\in [0,T)\times \R^2$
      \begin{align*}
      \Phi_{s,T}(x):= \lim_{t\uparrow T} \Phi_{s,t}(x)=\Phi(x)
   \end{align*}
   and
$t \mapsto k_t$ satisfies the identity 
   \begin{align}\label{eq:defv}
       \frac{1}{T}\int_0^T k_s^2 ds=\frac{v^2}{\sigma^2 T} \ .
   \end{align}

\end{lemma}

Given these functions, for any given continuous $\Fc$-semimartingale $\widetilde{Y}$, which is a candidate for an aggregate demand, define $\chi = \chi(\widetilde{Y})$ by the linear Volterra equation
$$\chi_t(\widetilde{Y})=\gamma\sigma^2\int_0^t \left( p_s \chi_s(\widetilde{Y}) +q_s \right)ds \ + \ \widetilde{Y}_t$$
and the candidate equilibrium pricing rule 
\begin{align*}
    \fP^*_t(\widetilde{Y}) := p_t \chi_t(\widetilde{Y}) + q_t \ .
\end{align*}

The linear structure of the pricing rule $\fP^*$ in $\widetilde Y$ is due to the Gaussianity of the distribution and justified in Section \ref{s.PricingRule} where we provide necessary conditions on the equilibrium. This section also justifies the expressions for $p,q$ by a martingality argument. The expressions of the other deterministic functions are more challenging to find. The equality \eqref{eq:defv} is a consequence of a fixed point approach in filtering problems between two probability measures that we introduce, and it pinpoints the value of $v$ in Lemma \ref{lem:v}. 
We now provide our main existence result, whose proof can be found in Section \ref{s.ProofMain}.
\begin{thm}[Existence of Equilibrium]
\label{thm:main}
    There exists a unique solution $(\chi^*,Z^\Q)$ to the system of forward stochastic differential equations
    \begin{align}
\label{eq:chi}  \chi^*_t&= \int_0^t [\gamma\sigma^2 (p_s \chi^*_s+q_s)+\sigma^2 \left(\rb_s -e_1\right)^\top \Sigma_s^{-1}((Z_s^\Q-\beta, \Xi )^\top-\mu^\P_s)]ds+Z_t \ , \\ 
Z_t^\Q &=  \int_0^t\gamma\sigma^2  (p_s \chi_s^*+q_s)ds+Z_t \ , \nonumber
\end{align}
where $\mu^\P$ is given by 
\begin{align}
    \label{eq:expmup}
        \mu_t^\P =  \ab_t +\rb_t \chi^*_t+\gamma\sigma^2\int_0^t (p_s\chi^*_s+q_s) \Phi \left(e_1-\rb_s \right) ds \ .
\end{align}
Define 
\begin{align}\label{eq:gausscontrol}
dX^*_t=\sigma^2 \left(\rb_t -e_1\right)^\top \Sigma_t^{-1}((Z_t^\Q-\beta, \Xi )^\top-\mu^\P_t)dt \ ,
\end{align}
and for all $\tilde Y\in\Sc(\Fc)$, $t\in [0,T]$,
$$\fP^*_t( \tilde Y ):=p_t \chi_t(\tilde Y)+q_t \ , $$
where $\chi(\tilde Y)$ solves
$$\chi_t(\tilde Y)=\int_0^t \gamma\sigma^2 (p_s \chi_s(\tilde Y)+q_s)ds+\tilde Y_t \ ,\mbox{ for }t\in [0,T] \ .$$
Then, $(X^*, \fP^*)$ forms an equilibrium where the inconspicuous trading property
\begin{align}\label{eq:incons}
    \E\Big[\frac{dX^*_t}{dt}|\Fc^{X^*+Z}_t\Big]=0
\end{align}
holds and conditionally on $\Fc^{X^*+Z}_t$, and $(Z_t^\Q-\beta, \Xi )^\top $ is $\Nc(\mu_t^\P,\Sigma_t)$.

\end{thm}

\begin{rmk}
(i) We define (for given trading strategy $X$) the measure 
$\Q$ by 
\begin{align*}
       \left.\frac{d\Q}{d\P}\right|_{\Fc_T}
 = & \ \exp\left( -\gamma\int_0^T  P_s dZ_s-\gamma\int_0^T \frac{\gamma \sigma^2 P_s^2}{2}ds\right) 
 \end{align*}
 where $P_s=\fP^*_s(X+Z)$ is the price process. 
As shown in Proposition \ref{p:opt1}, this change of measure removes the drifts in the utility maximization problem of the informed trader and under this probability the informed trader maximizes a $\Q$ utility from terminal wealth only as in Eq.\eqref{eq:uti2}. Similar to the classical Kyle's model \cite{kyle,back1992,back2020optimal,first_paper,ccd}, a simple optimality criterion for the informed trader is to bring the state at terminal time to a value given by the private information, see Eq.\eqref{eq:optimal_utility}. 

(ii) The matrix $\Sigma_t$ continuously interpolates the covariance matrix of $(-\beta, \Xi )^\top $ and the singular matrix $\Sigma_T=\frac{(\sigma^2 T +\sigma_\beta^2)\sigma^2_\Xi}{\sigma_{e}^2 }
     \vb\vb^\top$ which is the $\Fc^{X^*+Z}_T$ conditional covariance of $(Z_T^\Q-\beta, \Xi )^\top $. In fact, at time $T-$, the $\Fc^{X^*+Z}_T$ conditional variance of $(Z_T^\Q-\beta, \Xi )\wb$ is $0$ whereas the conditional variance of $(Z_T^\Q-\beta, \Xi )\vb$ is $\frac{(\sigma^2 T +\sigma_\beta^2)\sigma^2_\Xi}{\sigma_{e}^2 }
     |\vb|^4$ . Thus, for $\vare \gamma>0$, not all information on $\Xi$ is injected to the market at time $T-$ and the market maker only knows a linear combination of $\Xi$ and $Z_T^\Q$.

(iii) In many equilibria, the informed trader brings the total demand (or a process related to total demand) to a terminal total demand given by her private information. For example, in \cite{back1992}, for some function $F$, one has 
$$X^*_T=Y_T-Z_T=F(\Xi)-Z_T,$$
which means that the informed trader has to be counterpart to the rest of the market and absorbs all the noise trading. Similarly, with risk-aversion, \cite{cho,bose2020kyle,bose2021multidimensional} show that the optimal trading strategy of the informed trader is to enforce $\chi_T^*=F(\Xi)$.
In our framework, such equalities do not hold and the terminal value of $\chi^*_T$ is random. 

\end{rmk}
    

\subsection{Properties of the equilibrium}
The proof of the next theorem can be found in Section \ref{proof-3}.
\begin{thm}[Properties of the Equilibrium]
\label{thm:prop} 
${}$

    (1) In equilibrium, the price satisfies 
    \begin{align}
        \label{eq:pricedyn}
        dP_t=d\fP^*_t(\chi^*)=\frac{dY_t}{\Lambda^{-1}-\gamma\sigma^2 (T-t)}, \quad t \in [0, T)
    \end{align}
and the terminal price is 
$$P_T=\fP^*_T(\chi^*)= \Xi -\gamma \vare (X^*_T+\beta) \ .$$
The sensitivity of the price in the total demand (Kyle's Lambda) is 
\begin{align}\label{eq:kyleslambda}
    p_t := \frac{1}{\Lambda^{-1}-\gamma\sigma^2 (T-t)} \ . 
\end{align}

        (2) All other parameters being fixed, as $\gamma \downarrow 0$, $$\Lambda\uparrow \frac{\sigma_\Xi}{\sigma\sqrt{T}}\mbox{ and }p_t\uparrow \frac{\sigma_\Xi}{\sigma\sqrt{T}} \ , $$
where the monotonicity is valid on a neighborhood of $0$.

         All other parameters being fixed, as $\vare \downarrow 0$, $$\Lambda\uparrow \frac{\sigma_\Xi}{ \Xi }\mbox{ and }p_t\uparrow\frac{1}{\left( 
     \frac{\sigma_\Xi}{ \Xi }
     \right)^{-1}-\gamma\sigma^2 (T-t)} \ , $$
where the monotonicity is valid on a neighborhood of $0$ and $ \Xi $ is the solution to 
        $$ \frac{\Xi^2}{\sigma^2 T} 
   =
    \frac{1}
     {1 -  \frac{\sigma_\Xi\gamma\sigma^2 T}{ \Xi } } \ .
$$

        (3) The expected utility of the informed trader conditional to her private information $\{ \Xi = \xi ;\beta=b\}$ is  
\begin{align}\label{eq:utii}
   & -e^{-\gamma u(0,0)}  
   \E\left[\exp\left(-\gamma \Gamma^c(Z_T-b, \xi )\right)\right] \\ \nonumber
  =& -e^{-\gamma u(0,0)} 
   \frac{1}{\sqrt{1-\frac{v\Lambda}{\sigma_e}\vare \gamma^2 \sigma^2 T}}   
   \exp\Big(
    -\frac{\gamma}{2}(\frac{ \xi ^2}{\vare \gamma} + \frac{\vare \gamma \sigma_e}{v\Lambda} m^2)
    + \frac{\varepsilon \gamma^2}{2}\frac{ v \Lambda}{\sigma_e} 
\frac{\alpha^2}{1 - \frac{v\Lambda}{\sigma_e}\vare \gamma^2 \sigma^2 T}
   \Big) \ ,
\end{align}
where 
$\alpha 
 = b + m \frac{\sigma_e}{v\Lambda} - \frac{1}{\vare\gamma} \xi  
$,
$u$ solves the partial differential equation (PDE) \eqref{pde:chi1}, and 
$$
    \Gamma^c(z,\xi;\Nc(m,v^2))=\frac{1}{2}\left(-2\xi z -{\vare\gamma}z^2+\frac{v}{\sigma_{e}}(\xi-m_\Xi  +{\varepsilon \gamma} (z+m_\beta))^2+2m (\xi +{\varepsilon \gamma} z)\right) \ .
$$

        (4) The expected profit of the market maker is 
\begin{align}\label{eq:wealthMM}
  \frac{\vare}{\Lambda} \left(1-\vare\gamma \frac{\sigma^2T+\sigma^2_\beta}{\sigma_{e} v}\right)\ln(1-\gamma\Lambda\sigma^2 T)<0 \ .
\end{align}
    
\end{thm}

\begin{rmk}
    i) The monotonicity in $\gamma$ means that when the informed trader is less risk-averse, the market maker adjusts the prices more aggressively with the arrival of new orders. 
    This aligns with \cite{cho,bose2020kyle}, which noted that this effect arises because a more risk-averse informed trader, due to her accumulated position, is less likely to fully exploit her private information.
    As a consequence, the expected loss of the market maker decreases as $\gamma$ increases and the market can afford to charge a smaller transaction cost, i.e. a smaller Kyle's Lambda. 

    (ii) In the equilibrium we construct, the process $\chi_t^*$ is a path-dependent functional of $Y$ and the equilibrium price is a Markov function of this process. 

    (iii) However, the equilibrium strategy $X^*$ is significantly more complicated than the pricing rule. Indeed, in order to construct the equilibrium, we first postulate that it satisfies the inconspicuous trading property \eqref{eq:incons}. In our Gaussian framework, this property requires us to keep track of the conditional expectation of $(Z_t^\Q-\beta, \Xi )$ which is path-dependent in $\chi$.

\end{rmk}

\section{Necessary conditions on Markovian equilibria}
\label{s.PricingRule}

In this section we follow the methodology of \cite{cho} to obtain necessary conditions on inconspicuous equilibria. In particular, we see the natural appearance of the state variable $\chi$ and the measure $\Q$ for example. In subsequent sections, we use these necessary conditions to construct the candidate equilibrium strategies of each agents and prove the existence of the equilibrium. Although these steps are not needed for the proof of the main theorems, we below they are informative on the obtaining the structure of the equilibrium.
Thus, we assume the existence of an equilibrium, where 
$$ dX^*_t=\theta_t dt \ $$
and satisfying the inconspicuousness condition $$\E[\theta_t|\Fc_t^Y]=0 \ .$$

We first focus on the path-dependence of the pricing rule. 
If the informed trader is risk neutral, it is known that $Y_t$ is the relevant state variable and both $\theta$ and the price are functions of $Y_t$ (see \cite{kyle, back1992}). However, in the case the informed trader is risk averse, it is known that $Y$ does not carry all the information needed to construct the equilibrium (see \cite{cho,bose2020kyle,bose2021multidimensional}). In \cite{cho}, it is shown that one needs to introduce a new state variable 
$$d\xi_t=\lambda_t dY_t$$
and in \cite{bose2020kyle,bose2021multidimensional} it is assumed that 
\begin{align*}
d\xi_t=\frac{dY_t}{\partial_\xi \chi(t,\xi)}
\end{align*}
for some function $\chi(t,\xi)$. Then, different ingredients of the equilibrium can be written as functions of the current value of $(t,\xi_t)$ and the dependence of $\xi$ on the path of $Y$ carries all the path-dependence in the problem. We note that by applying It\^o's formula to $\chi(t,\xi_t)$, one can show that in all these papers, there exists a new state variable 
$\chi_t$ that 
carries all the path-dependence of the problem and it has the dynamics  
\begin{align}\label{eq:defchi} 
\notag
d\chi_t &= dY_t + B(t,\chi_t)dt \\
\chi_0&=0
\end{align}
for some function $B$. 
The following Proposition shows that in any equilibrium where there is such a state variable satisfying \eqref{eq:defchi}, one has very strong conclusion on the pricing rule and an optimality criterion for the informed trader. Note that we study such equilibria because the equilibrium we construct has this form. Define the surplus function 
\begin{align}\label{eq:surplus}
        \Sc\left( x,\xi \right)
        = \xi x - g \left(x\right)=\xi x - \frac{\varepsilon \gamma}{2} x^2  \ . 
\end{align}

\begin{proposition}\label{p:opt1}
    Let $(X^*,\fP^*)$ be an equilibrium such that there exist continuous functions $P,B:[0,T]\times \R\mapsto \R$ such that the pricing rule $\fP^*$ and the absolutely continuous trading strategy  $X^*$ satisfy 
    $\fP^*_t(\tilde Y)=P(t,\chi_t(\tilde Y))$, where
    $$
        \chi_t(\tilde Y)= \int_0^t B(s,\chi_s(\tilde Y))ds+\tilde Y_t \quad \mbox{ and } \quad \E\left[\frac{dX^*_t}{dt}|\Fc_t^{X^*+Z}\right]=0\mbox{ for all }t\in[0,T] \ .
    $$
    Denote $\Gamma(\chi):=\int_0^\chi P(T,s)ds$, and \begin{align}\label{eq:gammac}
    \Gamma^c(z,\xi):=\sup_{x\in \R}\left[\Sc (x-z, \xi)-\Gamma(x) \right] \ . 
    \end{align}
    Assume for all $z,\xi$, \eqref{eq:gammac} admits an optimizer and that the parabolic PDE 
    \begin{align}\label{pde:chi1}0 &=   \ \pa_t u+\frac{\sigma^2}{2}\pa_{\chi\chi}u+\frac{1}{2}\gamma\sigma^2(\pa_\chi u)^2 \\
    u(T,\chi)&=\Gamma(\chi) \notag
    \end{align}
    admits a smooth solution $u \in C^{1,2}([0,T]\times \R)$.
    For any admissible trading strategy $X$, and $t\geq 0$, define the probability measure $\Q$ by
\begin{align*}
       \left.\frac{d\Q}{d\P}\right|_{\Fc_t}
:= & \ \exp\left( -\gamma\int_0^t  P_s dZ_s-\gamma\int_0^t \frac{\gamma \sigma^2 P_s^2}{2}ds\right) \ , 
\end{align*}
where $P_s=P(s,\chi_s(X+ Z))$.
Then, 
\begin{align}\label{eq:pathdep}
        B(t,\chi)=\gamma\sigma^2 P(t,\chi)=\gamma\sigma^2 \pa_\chi u(t,\chi) \ ,
    \end{align} 
the process 
\begin{align*}
Z_t^\Q := & \ Z_t+\gamma\sigma^2 \int_0^tP_sds
\end{align*} 
is a $(\Q,\Fc)$-Brownian motion, and 
\begin{align}\label{eq:eqlaw}
\Lc^\Q\left(  Z_T^\Q-\beta,  \Xi  \right)    = \mu \otimes \nu =\Lc\left(Z_T-\beta,  \Xi  \right) \ .
\end{align}
    Additionally, any choice of absolutely continuous admissible strategy $X$ ensuring the point wise optimality criterion
\begin{align}\label{eq:target}
\chi_T\in \argmax_{x \in \R}\left[ \Sc \left(x-\left(Z^\Q_T-\beta\right),  \Xi \right)-\Gamma(x) \right]
\end{align}
is optimal and leads to the optimal utility
 \begin{align}
 & -e^{-\gamma u(0,0)}  
   \E\left[\exp\left(-\gamma \Gamma^c(Z_T-\beta, \Xi )\right)| \Xi ,\beta\right] \ .
 \label{eq:optimal_utility}
\end{align}

\end{proposition}

\begin{rmk}
    (i) Although the new measure $\Q$ depends on the strategy $X$ through $\chi$, the optimal utility \eqref{eq:optimal_utility} is independent of $X$. This point is crucial to be able to use this proposition for an optimality criterion.  
    
   (ii) Since not all vector fields in multi-dimensions are gradient vector fields, the Proposition requires an additional condition if one wants to study this problem in the multi-dimensional case: the terminal pricing function $\chi\mapsto P(T,\chi)$ is a gradient vector field.

    (iii) This Proposition justifies the introduction of $\chi^*$ in \eqref{eq:chi}, where 
    $$dY_t=\sigma^2 (\rb_t-e_1)^\top \Sigma_t^{-1}((Z_t^\Q-\beta, \Xi )^\top-\mu^\P_t)dt+dZ_t$$
    is the equilibrium total demand. 
\end{rmk}

\begin{proof}[Proof of Proposition \ref{p:opt1}]
    We split the proof of the Proposition in two parts. We first show the property of the pricing rule in \eqref{eq:pathdep}. Then, we show the optimality criterion \eqref{eq:target} and the optimal utility \eqref{eq:optimal_utility}.

\medskip

{\it Step 1: Proof of \eqref{eq:pathdep}}. 
Since $X^*$ is optimal, it is a fortiori optimal among all absolutely continuous strategies. Fix an absolutely continuous trading strategy $X$ admissible for $\fP^*$. Without loss of generality, we assume that it can be written as 
$dX_t= \theta_t dt$ for some $\Fc$-adapted process $\theta$ and 
\begin{align}\label{eq:evchi}
    d\chi_s=B(s,\chi_s)ds+dX_s+dZ_s=(B(s,\chi_s)+\theta_s)ds +\sigma dB_s \ .
\end{align}
Following the optimization problem from \eqref{eq:problem}, we assume that the informed trader starts at time $t$ from position $X_t+\beta=x$ in stock and $0$ in cash so that
the wealth of the informed trader can be written as
$$ \Wc:=(X_T+\beta)( \Xi -G) - \int_t^T  P(s,\chi_s) \theta_sds \ .
$$
And we introduce the optimal control problem
 \begin{align}
 \label{eq:def_J}
 J(t,\chi,x) = \ \sup_{\theta} e^{\gamma  \Xi  x} \E\left[-\exp\left(-\gamma\Wc \right) \ | \ \Fc_t, \chi_t = \chi, X_t+\beta=x\right] \ , 
\end{align}
where the state is $(\chi_t,X_t+\beta)$ and evolves according to \eqref{eq:evchi} and 
$$d(X_s+\beta)=\theta_sds$$
given control $\theta$. 
By the independence of $G$, we can write the utility as 
\begin{align*}
     &\E\left[-\exp\left(\frac{\gamma^2 \varepsilon (X_T+\beta)^2}{2}+ \gamma\int_t^T  (P(s,\chi_s)- \Xi ) \theta_sds \right) \ | \ \Fc_t, \chi_t = \chi, X_t+\beta=x\right]
\end{align*}
and we obtain the following dynamic programming equation for $J$
\begin{align*}
    &\pa_t J+\frac{\pa_{\chi\chi} J\sigma^2}{2}+\pa_\chi J B+\sup_{\theta\in \R}\{\theta (\pa_\chi J+\pa_xJ +\gamma(P- \Xi ))\}=0 \ . 
\end{align*}
The final condition of $J$ is
$J(T,\chi,x)=-e^{\frac{\gamma^2 \varepsilon x^2}{2}}$.

Since the control variable $\theta$ is unconstrained, the optimization on the variable $\theta$ is over all of $\R$. Necessarily, we obtain the pair of the equations 
\begin{align*}
&\pa_t J +\frac{\sigma^2}{2}\pa_{\chi\chi} J + B\pa_\chi J =0  \ , \\
&\pa_\chi J+\pa_xJ +\gamma(P- \Xi ) = 0 \ , 
\end{align*}
or, equivalently,
\begin{align}
\label{eq:lnj}
&  \pa_t \ln |J|
 + \frac{\sigma^2}{2} \pa_{\chi\chi} \ln|J|
 + \frac{\sigma^2}{2} \left( \pa_\chi \ln |J| \right)^2
 + B \pa_\chi \ln|J| = 0 \ , \\
&
\label{eq:lnj_2}
P= \Xi -\frac{ \partial_\chi \ln |J| + \partial_x \ln |J| }{\gamma} \ .
\end{align}
Thus, differentiating \eqref{eq:lnj} in each variable $\chi$ and $x$, while setting $V= \pa_\chi \ln |J|$ and $\tilde V= \pa_x \ln |J|$, we obtain
\begin{align*}
\pa_tV+\frac{\sigma^2}{2} V_{\chi\chi}+ V_\chi (B+\sigma^2V)+VB_\chi &=0 \ , \\
\pa_t\tilde V+\frac{\sigma^2}{2} \tilde V_{\chi\chi}+ \tilde V_\chi (B+\sigma^2V) &=0 \ .
\end{align*}
Thus, adding these two expressions and using \eqref{eq:lnj_2}, we have
\begin{align*}
\pa_t P+\frac{\sigma^2}{2} \pa_{\chi\chi} P + \pa_\chi P (B+\sigma^2V)-\frac{V \pa_\chi B}{\gamma} &=0 \ .
\end{align*}
The inconspicuousness condition 
$\E\left[\frac{dX^*_t}{dt}|\Fc_t^{X^*+Z}\right]=0$ implies that 
$\chi_t(X^*+Z)-\int_0^t B(s,\chi_s(X^*+Z))ds=X^*_t+Z_t$ is $\Fc^{X^*+Z}$-martingale.  
Invoking the martingality condition for $P$, we obtain the following important necessary condition on $B$ for the martingality of $\fP^*_t(X^*+Z)$
$$ \gamma\sigma^2 \pa_\chi P = \pa_\chi B \ , $$
which suggest that $B=\gamma\sigma^2 P$ up to the addition of a deterministic function of time $f = f(t)$. So, we have 
$$d\chi_t=dY_t+\gamma\sigma^2 (P(t,\chi_t)-f(t))dt \ .$$
Since it does not have informational value, we take $f=0$ and obtain
$$d\chi_t=dY_t+\gamma\sigma^2 P(t,\chi_t)dt \ .$$

Injecting the value of $B$, the martingality condition for $P$ is 
$$ \pa_t P+\frac{\sigma^2}{2} \pa_{\chi\chi} P + \gamma\sigma^2 P \pa_\chi P=0 \ .
$$
We now integrate this equality in space. 
If $P(t,\chi)$ is the gradient of some function  $u(t,\chi)$, then this function has to solve
\begin{align*}
0 = & \  \pa_t u+\frac{\sigma^2}{2}\pa_{\chi\chi} u +\frac{1}{2}\gamma\sigma^2 (\pa_\chi u )^2\\
  = &  \ \pa_t u+\frac{\sigma^2}{2}\pa_{\chi\chi} u +\frac{1}{2}\gamma\sigma^2P^2
\end{align*}
which concludes the proof of \eqref{pde:chi1}-\eqref{eq:pathdep}.

\medskip

{\it Step 2: Joint law of  $ \Xi $ and $Z_T^\Q-\beta$. }
 The change of measure is 
\begin{align*}
       \left.\frac{d\Q}{d\P}\right|_{\Fc_t}
 = & \ \exp\left( -\gamma\int_0^t  P_s dZ_s-\gamma\int_0^t \frac{\gamma \sigma^2 P_s^2}{2}ds\right) \\
 = & \ \exp\left( -\gamma \left( u(t, \chi_t) - u(0,0) \right) + \gamma \int_0^t P(s,\chi_s) dX_s \right) \\
 = & \ \exp\left( -\gamma \left( u(t, \chi_t) - u(0,0) \right) + \gamma \int_0^t \partial_\chi u(s,\chi_s) dX_s \right) \ .
\end{align*}
Thanks to Girsanov, we have indeed a $\sigma^2$-Brownian motion in the form of
\begin{align*}
Z_t^\Q := &  \ Z_t+\gamma\sigma^2 \int_0^tP(s,\chi_s)ds \ .
\end{align*} 

Now, let us prove Eq. \eqref{eq:eqlaw}. Indeed, since $P = \left( P_t \ ; \ t \geq 0 \right)$ is $\Fc^Y$-adapted, for any bounded continuous functions $f$ and $h$, we have
$$\E^\Q[f(Z_T^\Q)h( \Xi ,\beta)]=\E\left[\E[e^{-\gamma\int_0^T  P_s dZ_s-\gamma\int_0^T \frac{\gamma\sigma^2 P_s^2}{2}ds}f(Z_T^\Q)|\Fc_0]h( \Xi ,\beta)\right] \ .$$
By the properties of conditioning, 
$$\E[e^{-\gamma\int_0^T  P_s dZ_s-\gamma\int_0^T \frac{\gamma\sigma^2 P_s^2}{2}ds}f(Z_T^\Q)|\Fc_0]=\E[e^{-\gamma\int_0^T  P_s dZ_s-\gamma\int_0^T \frac{\gamma\sigma^2 P_s^2}{2}ds}|\Fc_0] \ \E^\Q[f(Z_T^\Q)|\Fc_0] \ .$$
Since $Z$ is a $(\Fc,\P)$-Brownian motion, $\E[e^{-\gamma\int_0^T  P_s dZ_s-\gamma\int_0^T \frac{\gamma\sigma^2 P_s^2}{2}ds}|\Fc_0]=1$. The fact that $Z^\Q$ is a $(\Fc,\Q)$-Brownian motion leads to
$$\E^\Q[f(Z_T^\Q)|\Fc_0]=\E^\Q[f(Z_T^\Q)]$$
which is a deterministic integral against Gaussian density. 
Thus, 
$$\E^\Q[f(Z_T^\Q)h( \Xi ,\beta)]=\E^\Q[f(Z_T^\Q)] \ \E[h( \Xi ,\beta)]$$
which concludes the proof of Eq. \eqref{eq:eqlaw}. 

\medskip

{\it Step 3: Proof of optimality criteria \eqref{eq:target}-\eqref{eq:optimal_utility}}.
Since $u$ is assumed to be smooth, by the Itô's formula, we have
\begin{align*}
    du(t,\chi_t)
= & \ \left( \pa_t u+\frac{\sigma^2}{2}\pa_{\chi\chi} u \right) dt
    + \pa_\chi u d\chi_t\notag\\
= & \ P(t,\chi_t)d\chi_t-\frac{1}{2}\gamma\sigma^2P^2(t,\chi_t)dt\notag\\
= & \ P(t,\chi_t)dX_t+P(t,\chi_t)dZ_t+\frac{1}{2}\gamma\sigma^2P^2(t,\chi_t)dt \ . 
\end{align*}
The change of measure is
$$\frac{d\Q}{d\P}=\exp\left( -\gamma\int_0^T  P_s dZ_s-\gamma\int_0^T \frac{\gamma \sigma^2 P_s^2}{2}ds\right)= \exp\left( -\gamma \left( u(T, \chi_T) - u(0,0) \right) + \gamma \int_0^T P_s dX_s \right)$$
so that the value function satisfies
 \begin{align}
  & e^{-\gamma  \Xi  \beta}J(0,0,\beta)\notag\\
= & \sup_{(X)}  \E\left[-e^{\gamma (g (X_T+\beta)- \Xi  (X_T+\beta))+\gamma\int_0^T  P_t dX_t}|\Fc_0\right]\notag\\
= & e^{-\gamma u(0,0)}  
    \sup_{(X)}\E\left[-e^{\gamma (g (X_T+\beta)- \Xi (X_T+\beta)+u(T,\chi_T))-\gamma\int_0^T  P_s dZ_s-\gamma\int_0^T \frac{\gamma \sigma^2 P_s^2}{2}ds}|\Fc_0\right]\notag\\
= & e^{-\gamma u(0,0)}  
    \sup_{(X)}\E\left[-e^{\gamma (g (X_T+\beta)- \Xi (X_T+\beta)+u(T,\chi_T)) }
                       \frac{d\Q}{d\P} \ | \ \Fc_0\right]\notag \\
= & e^{-\gamma u(0,0)} 
    \sup_{(X)}\E^\Q\left[-e^{\gamma (g (X_T+\beta)- \Xi  (X_T+\beta)+u(T,\chi_T)) }\ | \ \Fc_0\right] \ .\label{eq:uti2}
\end{align}
By the definition of $\chi$, we have 
\begin{align*}
   X_T+\beta=\chi_T-\int_0^T \gamma\sigma^2 P(s,\chi_s)ds-Z_T+\beta=\chi_T-Z_T^\Q+\beta
\end{align*}
and 
\begin{align*}
    \E^\Q\left[-e^{\gamma (g (X_T+\beta)- \Xi  (X_T+\beta)+u(T,\chi_T)) }\ | \ \Fc_0\right]&= \E^\Q\left[-e^{\gamma (g (\chi_T-Z_T^\Q+\beta)- \Xi  (\chi_T-Z_T^\Q+\beta)+u(T,\chi_T)) }\ | \ \Fc_0\right]\\
    &= \E^\Q\left[-e^{-\gamma (\Sc(\chi_T-(Z_T^\Q-\beta), \Xi )-u(T,\chi_T)) }\ | \ \Fc_0\right]\\
    &\leq \E^\Q\left[-e^{-\gamma (\sup_x\Sc(x-(Z_T^\Q-\beta), \Xi )-u(T,x)) }\ | \ \Fc_0\right]\\
    &\leq \E\left[-e^{-\gamma \Gamma^c(Z_T-\beta, \Xi ) }\ | \ \Fc_0\right]
\end{align*}
where we used \eqref{eq:eqlaw} at the last line. Thus, we obtain the upper bound \eqref{eq:optimal_utility}. For any strategy satisfying \eqref{eq:target}, all inequalities above are equalities and we obtain the optimality criterion.

\end{proof}


\section{Optimal transport}
\label{sect:StructureOptimizer}

The Proposition \ref{p:opt1} and Eq.\eqref{eq:target} show that the fundamental quantity to pinpoint in equilibrium is $\eta := \Lc^\Q( \chi_T)= \Nc(m, v^2)$. Indeed, the distribution of $(Z_T^\Q - \beta, \Xi )$ under $\Q$ and and the surplus function $\Sc$ are known. If we also know $\eta = \Lc^\Q( \chi_T)= \Nc(m, v^2)$, $\Gamma$ can be constructed by solving an optimal transport problem described below. We refer to \cite{villani,santambrogio2015optimal} for their proofs.

\begin{proposition}\label{prop:OT}
For any probability distribution $\eta$ on $\R$, denote $\Gamma(\chi;\eta)$ and $\Gamma^c(z,\xi;\eta)$ the Kantorovich potentials for the transport problem
\begin{align}\label{eq:optt}
\sup_{\pi \in \Pi(\mu\otimes \nu, \eta)}\iiint\Sc (\chi-z, \xi)d\chi dzd\xi \ , 
\end{align}
where the surplus function is defined is \eqref{eq:surplus}, and under $\pi$, the $\chi$ marginal is distributed according to $\eta$ and the law of $(z, \xi)$ is $\mu\otimes \nu$. 
By definition, these functions satisfy the duality relations
\begin{align*}
    \Gamma(\chi;\eta)+\Gamma^c(z,\xi;\eta)\geq \Sc (\chi-z, \xi)
\end{align*}
for all $(\chi,z,\xi)$ and there exists a mapping $I(z,\xi;\eta)$ so that  
$I\left( \cdot \ ;\eta \right)_\sharp (\mu\otimes \nu) = \eta$
and 
\begin{align*}
    \Gamma(I(z,\xi;\eta);\eta)+ \Gamma^c(z,\xi;\eta)= \Sc (I(z,\xi;\eta)-z, \xi)
\end{align*}
for $\mu\otimes \nu$ almost every $(z,\xi).$
Note that a consequence of this equality is the fact that for $\mu\otimes \nu$ almost every $(z,\xi)$,
\begin{align}\label{eq:defI}
I(z,\xi;\eta)\in \argmax\{\Sc (x-z, \xi)-\Gamma(x;\eta)\} \ .
\end{align}

We denote
\begin{align}\label{eq:defoptcoupling}
    \pi^*(d\chi,dz,d\xi;\eta):=((I(\cdot;\eta),\id)_\sharp(\mu\otimes \nu) )(d\chi,dz,d\xi)
\end{align}
and 
\begin{align*}
   \chi\mapsto \pi^*_\chi(dz,d\xi;\eta)
\end{align*}
the disintegration of this measure on the $\chi$ marginal. This means that 
$$ \pi^*(d\chi,dz,d\xi;\eta)=\pi^*_\chi(dz,d\xi;\eta)\eta(d\chi) \ .$$
\end{proposition}

Eq. \eqref{eq:defI} (and therefore \eqref{eq:target}) is in fact the main relation we use from optimal transport theory. This result says that if we take $\Gamma$ in \eqref{eq:target} to be the Kantorovich potential $\Gamma(\cdot;\eta)$, then the law of $\chi_T=I(Z_T^\Q - \beta, \Xi )$ under $\Q$  will indeed be $\eta$. This provides a way to pinpoint the solution to \eqref{pde:chi1} by finding its terminal condition as a Kantorovich potential if we know $\eta$.

If the measures $\mu$ and $\eta$ are Gaussian, we can explicitly compute these quantities as functions of the mean and the variance of $\eta$. 

\begin{lemma}
\label{lemma:gaussian_transportfinalcondition}
Assume that $\eta$ is normal with mean $m$ and variance $v^2>0$ and denote $\Lambda:=\frac{\sigma_{e}}{v}-\vare\gamma$. For the transport problem \eqref{eq:optt} with Gaussian marginals 
\begin{align}\label{Assump:Gaussian}
 \mu \otimes \nu 
 & = \mathcal{L}^\Q (Z_T^\Q-\beta,   \Xi  )
 = \Nc\left(-m_\beta,\sigma^2 T + \sigma^2_\beta \right) \otimes \Nc\left(m_\Xi ,\sigma^2_\Xi\right)  \ , \\
    \eta &= \mathcal{L}^{\mathbb{Q}}(\chi_T) 
        = \mathcal{N}(m, v^2) \ ,  \nonumber
\end{align}
the Kantorovich potentials are
\begin{align}
\label{eq:gausstransport}
\Gamma(\chi;\eta)&=\frac{1}{2}\Lambda(\chi-m)^2+(m_\Xi -\vare\gamma (m_\beta+m))(\chi-m) - \frac{\vare\gamma m^2}{2}  
\\
    \Gamma^c(z,\xi;\eta)&=\frac{1}{2}\left( - 2\xi z -{\vare\gamma}z^2+\frac{v}{\sigma_{e}}(\xi-m_\Xi  +{\varepsilon \gamma} (z+m_\beta))^2+2m (\xi +{\varepsilon \gamma} z)\right) \ \notag\end{align}
and the optimal transport maps are
\begin{equation} \label{eq:OTmapGaussian}
    I(z,\xi;\eta)=m+\frac{v}{\sigma_{e}}((\xi-m_\Xi ) +{\varepsilon \gamma} (z+m_\beta)) \ .
\end{equation}

If the joint law of $\left(\chi_T,  Z^\Q_T-\beta,  \Xi  \right) $  is $ \pi^*$, then under this coupling we have the equality in law
$$ \left( \begin{pmatrix} Z^\Q_T-\beta\\  \Xi \end{pmatrix} \ | \ \chi_T \right)
   \
   \eqlaw 
   \
   \frac{\chi_T-m}{\sigma_{e} v} 
   \ub
   +
   \begin{pmatrix}-m_\beta\\m_\Xi \end{pmatrix}
   +
   \Nc 
   \frac{\sigma_\Xi \sqrt{\sigma^2 T + \sigma_\beta^2} }{\sigma_{e}} 
   \vb
   \ 
$$
where $\Nc$ is an independent standard normal random variable. 
\end{lemma}

\begin{proof}
The optimal transport problem in hand is
\begin{align*}
&\sup_{\pi \in \Pi(\mu\otimes \nu, \eta)}\iiint\Sc (\chi-z, \xi)\pi(d\chi, dz,d\xi)\\
=&\sup_{\pi \in \Pi(\mu\otimes \nu, \eta)}\iiint \big( \xi (\chi-z) - \frac{\varepsilon \gamma}{2} \left(\chi-z\right)^2 \big) \ \pi(d\chi, dz,d\xi)\\
=&\iint \big( -z\xi  - \frac{\varepsilon \gamma}{2} z^2+m(\xi +{\varepsilon \gamma} z) \big) \ (\mu\otimes \nu)(dz,d\xi)\\
& \quad +\int \big( (m_\Xi -\vare\gamma m_\beta)(\chi-m) -\frac{\varepsilon \gamma}{2}\chi^2 \big) \ \eta(d\chi)\\
& \quad + \sup_{\pi \in \Pi(\mu\otimes \nu, \eta)}\iiint  (\chi-m)(\xi-m_\Xi  +{\varepsilon \gamma}(z+m_\beta)) 
\ \pi(d\chi, dz,d\xi)\notag \ . 
\end{align*}
The last term on the right is the optimal coupling between 
$$\Nc(m,v^2)\mbox{ and }\Nc(m_\Xi -\gamma \vare m_\beta,\sigma_{e}^2 ) \ .$$
The potentials are \begin{align*}
   \Gamma(\chi;\Nc(m,v^2))&=\frac{1}{2}\Lambda(\chi-m)^2+(m_\Xi -\vare\gamma (m_\beta+m))(\chi-m) - \frac{\vare\gamma m^2}{2} \\
    \Gamma^c(z,\xi;\Nc(m,v^2))&=\frac{1}{2}\left( - 2\xi z -{\vare\gamma}z^2+\frac{v}{\sigma_{e}}(\xi-m_\Xi  +{\varepsilon \gamma} (z+m_\beta))^2+2m (\xi +{\varepsilon \gamma} z)\right) \ .
\end{align*}
This is easily verified using the $\vare$-Young inequality. Because the optimal transport is linear for Gaussians, the optimizer is 
$$I(z,\xi;\Nc(m,v^2))=m+\frac{v}{\sigma_{e}}((\xi-m_\Xi ) +{\varepsilon \gamma} (z+m_\beta)) \ .$$

It remains to compute the conditional distribution $\pi^*_{\chi}(dz,d\xi;\Nc(m,v^2))$. Thanks to \eqref{eq:defoptcoupling}, we know that under $\pi^*$, 
$$ \chi_T = \ I(Z^\Q_T-\beta, \Xi ;\Nc(m,v^2))
          = \ m+\frac{v}{\sigma_{e}}( \Xi -m_\Xi  +{\varepsilon \gamma} (Z^\Q_T-\beta+m_\beta)) 
          \eqlaw \ \Nc(m, v^2) \ .$$
As such, the vector $\left( Z^\Q_T-\beta, \Xi ,\chi_T \right)^\top$ is a Gaussian vector with mean $(-m_\beta, m_\Xi , m)^\top$ and (degenerate) covariance matrix 
$$
\Sigma_{ext}
=
\begin{bmatrix}
\sigma^2 T +\sigma_\beta^2&     0    & \frac{v\varepsilon \gamma}{\sigma_{e}} (\sigma^2 T +\sigma_\beta^2)\\
0  &       \sigma_\Xi^2 &  \frac{v\sigma_\Xi^2}{\sigma_{e}} \\
\frac{v\varepsilon \gamma}{\sigma_{e}} (\sigma^2 T +\sigma_\beta^2)&\frac{v\sigma_\Xi^2}{\sigma_{e}} & v^2 
\end{bmatrix}
\ .
$$

The Schur complement formula states that $\left( (Z^\Q_T-\beta, \Xi )^\top \ | \  \chi_T \right)$ is Gaussian with explicit mean and covariance using block operations. More precisely, the mean is
\begin{align*}
    \E^\Q\left[ (Z^\Q_T-\beta, \Xi )^\top \ | \ \chi_T \right]
= & \ \E^\Q\left[ (Z^\Q_T-\beta, \Xi )^\top \right]
    + 
      \frac{1}{\sigma_{e} v} \left( \chi_T - m \right) \ub \\
= & \ \frac{\chi_T - m}{ \sigma_{e} v}
      \ub +\begin{pmatrix}-m_\beta\\m_\Xi \end{pmatrix}\ ,
\end{align*}
and the covariance is
\begin{align*}
    \Cov\left( (Z^\Q_T-\beta, \Xi )^\top \ | \ \chi_T \right)
= & \ \begin{bmatrix}
\sigma^2 T +\sigma_\beta^2 &         0    \\
0  &       \sigma_\Xi^2  \\
\end{bmatrix}
    - \frac{1}{\sigma_{e}^2 }
      \ub
      \ub^\top \ .\end{align*}
 Recalling that $\sigma_{e}^2=\sigma^2_\Xi+\vare^2\gamma^2( \sigma^2 T+\sigma^2_\beta)$, we see that
\begin{align*}
  & \Cov\left( (Z^\Q_T-\beta, \Xi )^\top \ | \ \chi_T \right) \\
= & \begin{bmatrix}
\sigma^2 T +\sigma_\beta^2& 0 \\ 0 & \sigma_\Xi^2
\end{bmatrix}
 - \frac{1}{\sigma_{e}^2 }
   \left( \varepsilon \gamma(\sigma^2 T +\sigma_\beta^2), \sigma_\Xi^2 \right)^\top
   \left( \varepsilon \gamma(\sigma^2 T +\sigma_\beta^2), \sigma_\Xi^2 \right)\\
= & \frac{1}{\sigma_{e}^2}
 \begin{bmatrix}
 (\sigma^2_\Xi+\vare^2\gamma^2(\sigma^2 T +\sigma_\beta^2)) (\sigma^2 T +\sigma_\beta^2) & 0 \\
 0 & (\sigma^2_\Xi+\vare^2\gamma^2 (\sigma^2 T +\sigma_\beta^2)) \sigma_\Xi^2
 \end{bmatrix}\\
 &-\frac{1}{\sigma_{e}^2}
 \begin{bmatrix}
 (\varepsilon \gamma (\sigma^2 T +\sigma_\beta^2))^2 & \varepsilon \gamma (\sigma^2 T +\sigma_\beta^2) \sigma_\Xi^2 \\
 \varepsilon \gamma (\sigma^2 T +\sigma_\beta^2) \sigma_\Xi^2 & (\sigma_\Xi^2)^2
 \end{bmatrix} 
 \\
= & \frac{1}{\sigma_{e}^2 }
 \begin{bmatrix}
 \sigma^2_\Xi (\sigma^2 T +\sigma_\beta^2) & -\varepsilon \gamma (\sigma^2 T +\sigma_\beta^2) \sigma_\Xi^2 \\
     -\varepsilon \gamma (\sigma^2 T +\sigma_\beta^2) \sigma_\Xi^2 & \vare^2\gamma^2(\sigma^2 T +\sigma_\beta^2)\sigma_\Xi^2
     \end{bmatrix}\\
 = & \frac{(\sigma^2 T +\sigma_\beta^2)\sigma^2_\Xi}{\sigma_{e}^2 }
     \begin{bmatrix}
     1 & -\varepsilon \gamma \\
     -\varepsilon \gamma & \vare^2\gamma^2
     \end{bmatrix}= \frac{(\sigma^2 T +\sigma_\beta^2)\sigma^2_\Xi}{\sigma_{e}^2 }
     \vb\vb^\top
 \end{align*}
which yields the result.

\end{proof}

\begin{rmk}
    In the remainder of the section, we take $\Gamma(x)=\Gamma(x;\Nc(m,v^2))$, $I(z,\xi)=I(z,\xi;\Nc(m,v^2))$ and drop the dependence of these functions on $m,v.$
\end{rmk}

The next Lemma states that if $\eta=\Nc(m,v^2)$ is given,  Proposition \ref{p:opt1} fully determines the pricing rule of the market maker. 
\begin{lemma}
\label{lemma:gaussian_solvability_price}
Assuming $m,v$ are given, the maps $u(t,\chi)$ and $P(t,\chi)$ are
$$ u(t,\chi)
   =
   \frac{p_t}{2}\chi^2 + q_t \chi + s_t \ ,
   \quad \quad
   P(t,\chi)= p_t\chi  + q_t \ , $$
where
{
\begin{align*}
p_t & :=\frac{1}{\Lambda^{-1}-{\gamma\sigma^2}(T-t)} \ , \quad 
q_t  := \frac{m_\Xi -\vare\gamma m_\beta - \frac{\sigma_e}{v}m}{1-\Lambda{\gamma\sigma^2}(T-t)} \ , \\
s_t & := \frac{1}{2\gamma} \log \left( 1-\Lambda{\gamma\sigma^2}(T-t)\right) 
			 - \frac{(m_\Xi -\vare\gamma m_\beta - \frac{\sigma_e}{v}m)^2}{ 2\Lambda} \left(
								\frac{1}{1 -  \Lambda{\gamma\sigma^2}(T-t)} -1
								\right)\\
         & \quad \quad
         +\frac{\sigma_e m^2}{2v} - m(m_\Xi -\vare\gamma m_\beta)
\ .
\end{align*}
}
In particular, for any inconspicuous strategy, 
$$ d\chi_t = dY_t + \gamma \sigma^2 P(t,\chi_t) dt \ $$
is an inhomogeneous Ornstein--Uhlenbeck process under $\P$ and the price satisfies the dynamics 
$$dP(t,\chi_t)=\frac{dY_t}{\Lambda^{-1}-{\gamma\sigma^2}(T-t)} \ .$$
\end{lemma}

\begin{proof}
Since $P(t,\chi) = \partial_\chi u(t,\chi)$ from Proposition \ref{p:opt1}, the first expression readily yields the second one by differentiation. In order to compute $u(t,\chi)$, recall that it is the solution to the PDE \eqref{pde:chi1}. Moreover, in the current Gaussian setup, the terminal condition is given from Eq. \eqref{eq:gausstransport}  by 
\begin{equation*}
    u(T, \chi) = \Gamma(\chi; \mathcal{L}^{\Q}(\chi_T) )= \frac{p_T}{2}\chi^2 + q_T \chi + s_T \ , 
\end{equation*}
where
\begin{align*}
p_T :=  \frac{\sigma_e}{v} - \vare\gamma=\Lambda \ ,
   \quad
   q_T := m_\Xi -\vare\gamma m_\beta - \frac{\sigma_e}{v}m \ , 
   \quad
   s_T : = \frac{\sigma_e m^2}{2v} - m(m_\Xi -\vare\gamma m_\beta) \ .
\end{align*} 
The PDE \eqref{pde:chi1} for $u$ is a well-behaved parabolic PDE and has obviously a unique solution. Therefore, it suffices to check that the proposed $u$ with $p_t, q_t, s_t$ satisfies the PDE \eqref{pde:chi1}.
Given a solution as in the statement, we have:
\begin{align*}
0 = & \ \partial_t u +\frac{\sigma^2}{2}(\partial_{\chi\chi} u +\gamma(\partial_\chi u )^2) \\
  = & \ \frac{p_t'}{2} \chi ^2+q_t'\chi +s_t'
        + \frac{\sigma^2}{2}\left( p_t + \gamma \left( p_t \chi + q_t \right)^2 \right) \\
  = & \ \frac{1}{2}\left( p_t' + \gamma \sigma^2 p_t^2
                   \right)\chi^2
        + \left( q_t' + \gamma\sigma^2 p_t q_t \right)\chi
        + s_t'+\frac{\sigma^2}{2}\left( p_t + \gamma q_t^2 \right) \ .
\end{align*}
Identifying terms, we necessarily have
{
$$  p'_t+{\gamma\sigma^2} p_t^2=0 \ ,
    \quad
q_t' + \gamma\sigma^2 p_t q_t =0 \ , 
	\quad
s_t'+\frac{\sigma^2}{2}\left( p_t + \gamma q_t^2 \right) =0 \ .
 $$
 }

As such, we can solve the ODEs successively. First, we have:
$$ p_t^{-1} = p_T^{-1} - \int_t^T d p_s^{-1} 
            = \Lambda^{-1} {+} \int_t^T \frac{p_s'}{p_s^2} ds
            = \Lambda^{-1} {-} \gamma \sigma^2 (T-t) \ .
$$
{
Then,
\begin{align*}
	\int_t^T p_s ds = & \int_0^{T-t} \frac{ds}{\Lambda^{-1} {-} \gamma  \sigma^2 s}\\
							= & -\frac{1}{\gamma \sigma^2} \left[ \log \big( \Lambda^{-1}-{\gamma\sigma^2}s \big) \right]_0^{T-t} \\
							= &  -\frac{1}{\gamma \sigma^2} \log \left( 1-\Lambda{\gamma\sigma^2}(T-t)\right) \ .
\end{align*}
Hence,
\begin{align*}
	q_t = q_T \exp\Big(\gamma\sigma^2 \int_t^T p_s ds \Big) 
			= \frac{m_\Xi -\vare\gamma m_\beta - \frac{\sigma_e}{v}m}{1-\Lambda{\gamma\sigma^2}(T-t)} \ .
\end{align*}
So, it gives
\begin{align*}
	\int_t^T q_s^2 ds = & q_T^2 \int_0^{T-t} \frac{ds}{( 1 - \Lambda{\gamma\sigma^2}s )^2} \\
								= & \frac{q_T^2}{ \Lambda{\gamma\sigma^2}} \left(
								\frac{1}{1 -  \Lambda{\gamma\sigma^2}(T-t)} -1
								\right) \ .
\end{align*}
Since $s_T = \frac{\sigma_e m^2}{2v} - m(m_\Xi -\vare\gamma m_\beta)$, we now have
\begin{align*}
	s_t = \frac{1}{2\gamma} \log \left( 1-\Lambda{\gamma\sigma^2}(T-t)\right) 
			- \frac{ q_T^2}{2\Lambda} 
        \left(\frac{1}{1 -\Lambda{\gamma\sigma^2}(T-t)} -1
								\right)
        + \frac{\sigma_e m^2}{2v} - m(m_\Xi -\vare\gamma m_\beta)
        \ .
\end{align*}
}

\end{proof}
We conclude this section by explaining how to obtain the expressions in \eqref{eq:defdet}.
Given the Proposition \ref{p:opt1}, we expect that the quantity that the market maker has to filter is $(Z_t^\Q - \beta, \Xi)$. If we know $m$ and $v$, then we know the $\Q$ distribution of $(Z_T^\Q - \beta, \Xi)$ conditional on $\Fc^\chi_T$ thanks to the Lemma \ref{lemma:gaussian_transportfinalcondition}. Assume that we have an equilibrium where the inconspicuous trading property holds. Then, the $\Q$ martingality of $(Z_t^\Q - \beta, \Xi)$ allows us to characterize the $\Fc^\chi_t$ conditional law of $(Z_t^\Q - \beta, \Xi)$ under $\Q$ as $\Nc(\mu^\Q_t,\Sigma_t)$, where
$$   \mu^\Q_t
 :=   \E^\Q\left[ (Z_t^\Q - \beta, \Xi)^\top \ | \ \Fc_t^\chi \right]
 =\begin{pmatrix}-m_\beta\\m_\Xi\end{pmatrix}
   +
   \frac{\int_0^t k_{s} dY^\Q_s}{\lambda v} 
   \ub
$$ 
and covariance 
$$ \Sigma_t= \Cov\left( (Z_t^\Q - \beta, \Xi)^\top \ | \ \Fc_t^\chi \right)
  =  \begin{bmatrix}
     \sigma^2 t + \sigma_\beta^2 & 0\\
    0           & \sigma^2_\Xi
   \end{bmatrix}
   -
   \frac{\sigma^2 \int_0^t k_{s}^2 ds}
        {\lambda^2 v^2}
   \ub \ub^* \ ,
$$
where $dY^\Q$ is the innovation process under $\Q$ and $k_s$ is $\Fc^\chi$-adapted that we need to determine. 
Thus, the optimal transport and martingale methods give us a parametric form for the conditional covariance of $(Z_t^\Q - \beta, \Xi)$ that only depends on the correct choice of $k_s.$ Starting from this ansatz on $\Sigma_t$, one can use \cite[Theorem 12.7]{liptser1977statistics} to find an expression for $X^*$. This forces us to study the filtering problems under both $\Q$ and $\P$. Then, we can find the correct expression for $k_s$ in \eqref{eq:defdet} by studying the conditional means $\mu^\Q_t, \mu^\P_t$ as in the Proof of Theorem \ref{thm:main} which can be found in Section \ref{s.ProofMain}.

\section{Proofs}
\label{s.appendix}


\subsection{Proof of Lemma \ref{lem:v}}
\label{proof-1}

The function 
$$x\mapsto F(x):= \frac{x^2}
     {1 - \gamma\sigma^2 T x } 
     - 
     \frac{2\vare }{\sigma^2 T } 
     \log \left( 1-\gamma\sigma^2 T x \right)$$
     is clearly increasing for $x\in [0,\frac{1}{\gamma \sigma^2T}]$ with 
     \begin{align*}
         F(0)=0,\,\lim_{x\uparrow\frac{1}{\gamma \sigma^2T}} F(x)=+\infty \ .
     \end{align*}
    Thus, there exists a unique root $x\in (0,\frac{1}{\gamma\sigma^2 T})$ to
      $F(x)=
       \frac{\sigma_\Xi^2}{\sigma^2 T}+\vare^2\gamma^2 \frac{\sigma_\beta^2}{\sigma^2 T}$. We choose $v$ so that 
       $$\frac{\sigma_{e}}{v}=x+\vare \gamma>0 \ .$$
       Then, the bounds on $\frac{v\vare\gamma}{\sigma_{e}}$ is a direct consequence of  $x\in (0,\frac{1}{\gamma\sigma^2 T})$ .

\subsection{Proof of Lemma \ref{lem:sigma}}
\label{proof-2}
We first prove \eqref{eq:defv}. For shorter notation denote
$$ g(v) := \gamma\sigma^2 T \Lambda= \gamma\sigma^2 T\left(\frac{\sigma_{e}}{v}-\vare\gamma \right)\ .$$
It is in fact simpler to compute $\frac{1}{T} \int_0^T (k_{t}^2 - 1) dt$ and add $1$ to the result. This computation is
\begin{align*}
  & \frac{1}{T} \int_0^T (k_{t}^2 - 1) dt= \int_0^1 \left[ 
             \left( 
             \frac{1-\vare \gamma \frac{v}{\sigma_{e}} 
                     g(v) u}
                  {1-g(v) u} \right)^2 - 1 \right] du =  \frac{1}{g(v)} 
    \int_0^{g(v)} \left[ 
             \left( 
             \frac{1-\vare \gamma \frac{v}{\sigma_{e}} 
                     u}
                  {1-u} \right)^2 - 1 \right] du \\
= & \frac{1}{g(v)} 
    \int_{1-g(v)}^1 \left[ 
             \left( 
             \frac{1-\vare \gamma \frac{v}{\sigma_{e}} 
                     (1-s)}
                  {s} \right)^2 - 1 \right] ds =  \frac{1}{g(v)} 
    \int_{1-g(v)}^1 \left[ 
             \left( 
             \frac{1-\vare \gamma \frac{v}{\sigma_{e}}}{s}
             + \vare \gamma \frac{v}{\sigma_{e}}
             \right)^2 - 1 \right] ds \\
= & \frac{1}{g(v)} 
    \int_{1-g(v)}^1 \left[ 
             \left( 
             \frac{1-\vare \gamma \frac{v}{\sigma_{e}}}{s}
             \right)^2
             +
             2
             \frac{1-\vare \gamma \frac{v}{\sigma_{e}}}{s}
             \vare \gamma \frac{v}{\sigma_{e}}
             +
             \left(
             \vare \gamma \frac{v}{\sigma_{e}}
             \right)^2 - 1 \right] ds \\
= &  \left( 
     1-\vare \gamma \frac{v}{\sigma_{e}}
     \right)^2
     \frac{1}{g(v)} 
     \left( -1 + \frac{1}{1-g(v)} \right)
     -
     2
     \left( 1-\vare \gamma \frac{v}{\sigma_{e}} \right)
     \vare \gamma \frac{v}{\sigma_{e}}
     \frac{1}{g(v)} 
     \log \left( 1-g(v) \right)\\
  & \quad 
     -
     \left( 1
     -
     \left(
     \vare \gamma \frac{v}{\sigma_{e}}
     \right)^2 \right) \\
= &  \left( 
     1-\vare \gamma \frac{v}{\sigma_{e}}
     \right)^2
     \frac{1}{1 - g(v)} 
     -
     2
     \left( 1-\vare \gamma \frac{v}{\sigma_{e}} \right)
     \vare \gamma \frac{v}{\sigma_{e}}
     \frac{1}{g(v)} 
     \log \left( 1-g(v) \right)
     -
     \left( 1
     -
     \left(
     \vare \gamma \frac{v}{\sigma_{e}}
     \right)^2 \right) \ .
\end{align*}
Thus, we have 
$$ \frac{1}{T} \int_0^T k_{t}^2  dt =
   \left( 
     1-\vare \gamma \frac{v}{\sigma_{e}}
     \right)^2
     \frac{1}{1 - g(v)} 
     -
     2
     \left( 1-\vare \gamma \frac{v}{\sigma_{e}} \right)
     \vare \gamma \frac{v}{\sigma_{e}}
     \frac{1}{g(v)} 
     \log \left( 1-g(v) \right)
     +
     \left(
     \vare \gamma \frac{v}{\sigma_{e}}
     \right)^2 \ .
$$
Injecting the definition of $v$ in Lemma \ref{lem:v}, we obtain that 
$$  \frac{1}{T} \int_0^T k_{t}^2  dt=\frac{v^2}{\sigma^2 T}  \ .
$$

We now prove the properties of $\Sigma_t.$ Given the definition of $f_t$, $\Sigma_t^{-1}$ can be computed using the Sherman-Morrison formula as 
\begin{align}\label{eq:sigma-1}
  & \Sigma_t^{-1}  =  \begin{bmatrix}
     \sigma^2 t + \sigma_\beta^2 & 0 \\
    0           & \sigma^2_\Xi
   \end{bmatrix}^{-1}
   +
   \frac{\sigma^2 \int_0^t k_{s}^2 ds}
        {\sigma_{e}^2 v^2}
   \frac{
   \begin{bmatrix}
     \sigma^2 t + \sigma_\beta^2 & 0 \\
    0           & \sigma^2_\Xi
   \end{bmatrix}^{-1}
   \ub
   \ub^\top
   \begin{bmatrix}
     \sigma^2 t + \sigma_\beta^2 & 0 \\
    0           & \sigma^2_\Xi
   \end{bmatrix}^{-1}
   }
   {
    1 - f_{t}
    } \ ,
\end{align}
if  the condition $1-f_t \neq 0$ holds.
We can compute $f_t$ as
\begin{align*}
f_{t} := 
   & 
   \frac{\sigma^2 \int_0^t k_{s}^2 ds}
        {\sigma_{e}^2 v^2}
   \ub^\top
   \begin{bmatrix}
     \sigma^2 t +  \sigma_\beta^2 & 0 \\
    0           & \sigma^2_\Xi
   \end{bmatrix}^{-1}
   \ub
   \\
= & 
   \frac{\sigma^2 \int_0^t k_{s}^2 ds}
        {\sigma_{e}^2 v^2}
      \left( 
   \frac{\gamma^2 \varepsilon^2 \left( \sigma^2 T + \sigma_\beta^2 \right)^2}{\sigma^2 t + \sigma_\beta^2} 
   +
    \frac{\sigma_\Xi^4}{\sigma^2_\Xi}
   \right)
   \\
= & 
   \frac{\sigma^2}
        {v^2}
    \int_0^t k_{s}^2 ds
      \left( 1+
   \frac{\sigma^2 (T-t )}{\sigma^2 t + \sigma_\beta^2} 
   \left( 1 - \frac{\sigma_\Xi^2}{\sigma_{e}^2} \right) 
   \right) \ .
\end{align*}
And, by integrating $k_s^2$, one can easily check that $f_t \neq 1$. So, $\Sigma_t$ is invertible with inverse given by \eqref{eq:sigma-1}.

Now, we check the limits. First, we have
\begin{align*}
1-f_{t}
= & 1- 
   \frac{\sigma^2}
        {v^2}
    \int_0^t k_{s}^2 ds
      \left( 
   1+
   \frac{\sigma^2 (T-t)}{\sigma^2 t + \sigma_\beta^2} 
   \left( 1 - \frac{\sigma_\Xi^2}{\sigma_{e}^2} \right) 
   \right) \\
= & 1+\left( 
    \frac{\sigma^2}
          {v^2}
    \int_t^T k_{s}^2 ds
    -1\right)
    \left( 
   1
   +
   \frac{\sigma^2 (T-t)}{\sigma^2 t + \sigma_\beta^2} 
   \left( 1 - \frac{\sigma_\Xi^2}{\sigma_{e}^2} \right) 
   \right) \\
= &  \frac{\sigma^2}{v^2} 
    \int_t^T  k_{s}^2 ds\left( 
   1
   +
   \frac{\sigma^2 (T-t)}{\sigma^2 t + \sigma_\beta^2} 
   \left( 1 - \frac{\sigma_\Xi^2}{\sigma_{e}^2} \right) 
   \right)-
   \frac{\sigma^2(T-t)}{\sigma^2 t + \sigma_\beta^2} 
   \left( 1 - \frac{\sigma_\Xi^2}{\sigma_{e}^2} \right) \ . 
\end{align*}
Integrating $k_s^2$, we obtain 
\begin{align*}
&\frac{\sigma_{e}^2(1-f_{t})}{(T-t)\sigma^2} 
= 
   -\vare^2\gamma^2\frac{\sigma^2 T + \sigma_\beta^2}{\sigma^2 t + \sigma_\beta^2}+\left(
   1
   +
   \frac{\sigma^2 (T-t)}{\sigma^2 t + \sigma_\beta^2} 
   \left( 1 - \frac{\sigma_\Xi^2}{\sigma_{e}^2} \right) 
   \right)\\
    & \quad \times \left(\vare^2\gamma^2-\frac{2\vare}{\sigma^2(T-t)}\log\left(1-\gamma\sigma^2\Lambda(T-t)\right)+\frac{\Lambda^2}{1-\gamma\sigma^2\Lambda(T-t)}\right)\\
   &=\frac{\sigma_\Xi^2}{\sigma^2 T + \sigma_\beta^2} 
   +\left(
   1
   +
   \frac{\sigma^2 (T-t)}{\sigma^2 t + \sigma_\beta^2} 
   \left( 1 - \frac{\sigma_\Xi^2}{\sigma_{e}^2} \right) 
   \right)\\
    &\times \left(-\frac{\sigma_\Xi^2}{\sigma^2 T + \sigma_\beta^2} -\frac{2\vare}{\sigma^2(T-t)}\log\left(1-\gamma\sigma^2\Lambda(T-t)\right)+\frac{\Lambda^2}{1-\gamma\sigma^2\Lambda(T-t)}\right) \ . 
\end{align*}

Additionally, thanks to \eqref{eq:defv}, we obtain that 
\begin{align}\label{eq:den_equivalent}
    \frac{1-f_t}{T-t}\to \frac{\sigma^2}{\sigma_{e}^2}\left(2\vare\gamma \Lambda+\Lambda^2\right)= \frac{\sigma^2}{\sigma_{e}^2}\left(\frac{\sigma_{e}^2}{v^2}-\vare^2\gamma^2\right)>0
\end{align}
and we have the equivalence at the singularity $t\to T$
\begin{align}
  & \wb^\top \Sigma_t^{-1} \wb\nonumber 
\sim_{t \rightarrow T} 
   \frac{
   \wb^\top\begin{bmatrix}
     \sigma^2 T + \sigma_\beta^2 & 0 \\
    0           & \sigma^2_\Xi
   \end{bmatrix}^{-1}
   \ub
   \ub^\top
   \begin{bmatrix}
     \sigma^2 T + \sigma_\beta^2 & 0 \\
    0           & \sigma^2_\Xi
   \end{bmatrix}^{-1}\wb
   }
   { (T-t)\sigma^2\left(\frac{\sigma_{e}^2}{v^2}-\vare^2\gamma^2\right) }
=    
   \frac{
   |
   \wb|^4
   }
   {  (T-t)\sigma^2\left(\frac{\sigma_{e}^2}{v^2}-\vare^2\gamma^2\right)}
\end{align}
where we used 
$$
\begin{bmatrix}
 \sigma^2 T + \sigma_\beta^2 & 0 \\
0           & \sigma^2_\Xi
\end{bmatrix}^{-1}
\ub
= 
\begin{bmatrix}
 \sigma^2 T + \sigma_\beta^2 & 0 \\
0           & \sigma^2_Xi
\end{bmatrix}^{-1}
\begin{bmatrix}
\vare\gamma(\sigma^2T+\sigma_\beta^2) \\ \sigma^2_\Xi
\end{bmatrix}
=
\begin{bmatrix}
\vare\gamma \\ 1
\end{bmatrix} 
=
\wb
\ . 
$$
Notice  $\wb$ is orthogonal to $\vb$.
To compute the other expansion in \eqref{eq:sigma-equivalent}, we use \eqref{eq:sigma-1}-\eqref{eq:den_equivalent} with the expansions  
\begin{align*}
   &\vb^\top\begin{bmatrix}
     \sigma^2 t + \sigma_\beta^2 & 0 \\
    0           & \sigma^2_\Xi
   \end{bmatrix}^{-1}
   \ub
   \ub^\top
   \begin{bmatrix}
     \sigma^2 t + \sigma_\beta^2 & 0 \\
    0           & \sigma^2_\Xi
   \end{bmatrix}^{-1}\vb= \left( \ub^\top
   \begin{bmatrix}
     \sigma^2 t + \sigma_\beta^2 & 0 \\
    0           & \sigma^2_\Xi
   \end{bmatrix}^{-1}\vb\right)^2\\
   &=\gamma^2 \vare^2\left( 
     \frac{\sigma^2 T + \sigma_\beta^2}{\sigma^2 t + \sigma_\beta^2} -
    1
   \right)^2=\gamma^2 \vare^2\sigma^4 
     \frac{(T-t)^2}{(\sigma^2 t + \sigma_\beta^2)^2} 
\end{align*}
and 
\begin{align*}
   &\wb^\top\begin{bmatrix}
     \sigma^2 t + \sigma_\beta^2 & 0 \\
    0           & \sigma^2_\Xi
   \end{bmatrix}^{-1}
   \ub
   \ub^\top
   \begin{bmatrix}
     \sigma^2 t + \sigma_\beta^2 & 0 \\
    0           & \sigma^2_\Xi
   \end{bmatrix}^{-1}\vb\sim_{t\uparrow T}\frac{2 (\vare\gamma)^3\sigma^2(T-t)}{\sigma^2T+\sigma^2_\beta} \ .
\end{align*}

Given the orthogonality of $\vb,\wb$ and \eqref{eq:defv}, the expression of $\Sigma_T$ is a direct consequence of the computation of 
\begin{align*}
    \wb^\top \Sigma_T\wb&=\wb^\top \begin{bmatrix}
     \sigma^2 T + \sigma_\beta^2 & 0\\
    0           & \sigma^2_\Xi
   \end{bmatrix}
   \wb-\frac{1}
        {\sigma_{e}^2 }\wb^\top
   \ub \ub^\top\wb=0\\
      \vb^\top \Sigma_T\vb&=\vb^\top \begin{bmatrix}
     \sigma^2 T + \sigma_\beta^2 & 0\\
    0           & \sigma^2_\Xi
   \end{bmatrix}
   \vb-\frac{1}
        {\sigma_{e}^2 }\vb^\top
   \ub \ub^\top\vb\\
   &=
     \sigma^2 T + \sigma_\beta^2 + \vare^2\gamma^2\sigma^2_\Xi
   -\frac{ \vare^2\gamma^2}
        {\sigma_{e}^2 }|\sigma^2T+\sigma_\beta^2-\sigma_\Xi^2|^2\\
   & 
   =\frac{ ( \vare^2\gamma^2(\sigma^2 T + \sigma_\beta^2) +\sigma^2_\Xi)(\sigma^2 T + \sigma_\beta^2 + \vare^2\gamma^2\sigma^2_\Xi)- \vare^2\gamma^2|\sigma^2T+\sigma_\beta^2-\sigma_\Xi^2|^2}
        {\sigma_{e}^2 }\\
        &=\frac{ (1+ \vare^2\gamma^2)^2(\sigma^2T+\sigma_\beta^2)\sigma_\Xi^2}
        {\sigma_{e}^2 }= \frac{(\sigma^2 T +\sigma_\beta^2)\sigma^2_\Xi}{\sigma_{e}^2 }
     \vb^\top\vb\vb^\top \vb
\end{align*}
which proves all the properties of $\Sigma$.

The invertibility of $\Sigma_t$ for $t<T$ easily shows the existence of $ \left( \Phi_{s,t} \ ; \ s \leq t < T \right)$. However, the flow $\Phi_{s,t}$ has a singularity at the final time $T$ that needs to be computed. Let $x_t = \Phi_{s, t}(x)$ for any $s< t$, which by definition satisfies
$ x_{s} = x$ and 
$$ dx_t = - \sigma^2 \rb_t \left(\rb_t - e_1 \right)^\top \Sigma_t^{-1} x_t dt \ .$$
Since $\frac{dx_t}{dt}$ is at direction $\ub$, we have that $x_t \in x + \ub \R$ for all $t\in [s,T)$.
Recalling the expression of $\Sigma_t^{-1}$, we have rational singularity from Eq. \eqref{eq:sigma-equivalent}:
$$ \frac{dx_t}{dt}
   \sim_{t \rightarrow T}
   - \frac{1}{\frac{\sigma_{e}^2}{v^2}-\vare^2\gamma^2} 
   \rb_t
   \frac{\left( \rb_t - e_1 \right)^\top \wb \wb^\top x_t}
   {T-t} 
   \sim_{t \rightarrow T}
   - \frac{1}{\frac{\sigma_{e}^2}{v^2}-\vare^2\gamma^2} 
   \frac{\ub}{\sigma_{e} v} 
   \frac{\left( \frac{\ub}{\sigma_{e} v}  - e_1 \right)^\top \wb \wb^\top x_t}
   {T-t} 
   \ .
$$
As such, 
\begin{align*}\frac{d(\wb^\top x_t)}{dt}&\sim_{t\rightarrow T} - \frac{1}{\frac{\sigma_{e}^2}{v^2}-\vare^2\gamma^2} 
   \frac{\wb^\top\ub}{\sigma_{e} v} 
   \frac{\left( \frac{\gamma^2\vare^2}{\sigma_{e} v}(\sigma^2 T+\sigma_\beta^2)-\gamma \vare +\frac{\sigma_\Xi^2}{\sigma_{e} v}\right)(\wb^\top x_t)}
   {T-t} \\
   &=- \frac{1}{\frac{\sigma_{e}^2}{v^2}-\vare^2\gamma^2} 
   \frac{\wb^\top\ub}{\sigma_{e} v} 
   \frac{\Lambda(\wb^\top x_t)}
   {T-t}
\end{align*}
where all the terms in front of $-\frac{\wb^\top x_t}{T-t}$ is positive. Thus, the flow of the ODE forces $\wb^\top x_t$ to revert to $0$ at time $T$ and $x_T$ is in the kernel of the linear form $\wb^\top$. Moreover, recall that $x_t \in x + \ub \R$. Therefore, $\Phi_{t,T}(x) = x + \ub y$, for some $y \in \R$ and
\begin{align*}
    \wb^\top\Phi_{t,T}(x)=\wb^\top \left( x + \ub y \right) = 0 \ ,
\end{align*} 
which is equivalent to
$$ y 
   = 
   - \frac{\wb^\top x}
          {\wb^\top \ub} \ .
   $$
This yields that $\Phi_{t,T}$ is independent of $t$ for $t<T$ and
$$ \Phi_{t,T}(x)
   = 
   x
   - \frac{\wb^\top x   }
          {\wb^\top \ub }
    \ub
          \ 
   $$
which concludes the proof of the lemma. 


\subsection{Proof of Theorem \ref{thm:main}}
\label{s.ProofMain}

\medskip

{\textbf{Step 1: Existence of solution to the linear system.}}
The existence of solution to the system \eqref{eq:chi}-\eqref{eq:expmup} is obvious for $t<T$ due to the linearity of the system and the invertibility of $\Sigma_t$ for $t<T$. The existence on $[0,T]$ will be shown below by showing the existence of a limit as $t\to T$. 
Injecting in \eqref{eq:expmup}, the expression of $\Phi$ and $\rb$, we obtain 
\begin{align*}
        \mu_t^\P &=  \ab_t +\rb_t \chi^*_t+\gamma\sigma^2\int_0^t (p_s\chi^*_s+q_s)  \left(e_1-\rb_s  -\frac{\wb^\top (e_1-\frac{k_s\ub}{\sigma_e v})}{\sigma_e^2 } \ub\right) ds\\
         &=  \ab_t +\frac{k_t\chi^*_t\ub}{\sigma_e v} +\gamma\sigma^2\int_0^t (p_s\chi^*_s+q_s)  \left(e_1  - \frac{\gamma \vare}{\sigma_e^2} \ub\right) ds \ . 
\end{align*}
Differentiating this equality and using the equalities 
    $ k_t' = \gamma \sigma^2 p_t ( \frac{v \varepsilon \gamma}{\sigma_e} -k_t )  $, $\ab'_t=-\gamma\sigma^2 q_t \left(\frac{k_t}{\sigma_e v} - \frac{\gamma \vare}{\sigma_e^2}\right)\ub$,
we obtain 
\begin{align*}
        d\mu_t^\P 
         &=  \ab'_tdt +\frac{k_t'\chi^*_t\ub}{\sigma_e v}dt+\left(\frac{k_t}{\sigma_e v} - \frac{\gamma \vare}{\sigma_e^2}\right)\gamma\sigma^2 (p_t\chi^*_t+q_t)\ub dt\\
         & \quad + \gamma\sigma^2(p_t\chi^*_t+q_t)  e_1  dt+\rb_t(dX^*_t +dZ_t)\\
        &= \gamma\sigma^2(p_t\chi^*_t+q_t)  e_1  dt+\rb_t(dX^*_t +dZ_t) \ .
\end{align*}
We also differentiate \eqref{eq:chi} to obtain
\begin{align}
\notag
    d\mu_t^\P&=\gamma\sigma^2 (p_t \chi^*_t +q_t) e_1 dt+\rb_t(dX^*_t +dZ_t)\label{eq:dynp}\\
    d\begin{pmatrix}
    Z^\Q_t-\beta\\ \Xi 
    \end{pmatrix}&=\gamma\sigma^2 (p_t \chi^*_t +q_t) e_1 dt+e_1 dZ_t 
\end{align}
so that 
\begin{align*}
 d\left(\begin{pmatrix}     Z^\Q_t-\beta\\ \Xi      \end{pmatrix}-\mu_t^\P\right)&=-\rb_t dX^*_t+\left(e_1-\rb_t \right) dZ_t\\
 &=-\rb_t \sigma^2 \left(\rb_t -e_1\right)^\top \Sigma_t^{-1}\left(\begin{pmatrix}
    Z^\Q_t-\beta\\ \Xi 
    \end{pmatrix}-\mu^\P_t\right)dt+\left(e_1-\rb_t \right) dZ_t \ .
\end{align*}
Recalling that $\vb,\wb$ are orthogonal, multiplying the equation above with $\vb$ and $\wb$, and using the expression of $\rb$ in \eqref{eq:defdet}, the processes defined by 
$$A^1_t:=\vb^\top\left(\begin{pmatrix}
    Z^\Q_t-\beta\\ \Xi 
    \end{pmatrix}-\mu_t^\P\right),
    \quad
    A^2_t:=\wb^\top\left(\begin{pmatrix}
    Z^\Q_t-\beta\\ \Xi 
    \end{pmatrix}-\mu_t^\P\right)$$
satisfy the linear equation in $(A^1,A^2)$, 
\begin{align*}
 dA^1_t=&-\frac{k_t\vb^\top \ub}{\sigma_{e} v}\sigma^2 \left(\frac{k_t\ub}{\sigma_{e} v}-e_1\right)^\top \Sigma_t^{-1}\frac{\vb}{|\vb|^2}A^1_tdt\\
 &-\frac{k_t\vb^\top \ub}{\sigma_{e} v}\sigma^2 \left(\frac{k_t\sigma_{e}}{ v}-\gamma\vare\right) \frac{\wb^\top \Sigma_t^{-1}\wb}{|\wb|^4}A^2_tdt\\
 &-\frac{k_t\vb^\top \ub}{\sigma_{e} v}\sigma^2 \left(\frac{k_t\ub^\top\vb}{\sigma_{e} v}-e_1^\top\vb\right)\frac{\vb^\top \Sigma_t^{-1}\wb}{|\wb|^4}A^2_tdt+\vb^\top\left(e_1-\rb_t\right) dZ_t \ , \\
  dA^2_t=&-\frac{k_t\wb^\top \ub}{\sigma_{e} v}\sigma^2 \left(\frac{k_t\ub}{\sigma_{e} v}-e_1\right)^\top \Sigma_t^{-1}\frac{\vb}{|\vb|^2}A^1_tdt\\
 &-\frac{k_t\wb^\top \ub}{\sigma_{e} v}\sigma^2 \left(\frac{k_t\sigma_{e}}{ v}-\gamma\vare\right) \frac{\wb^\top \Sigma_t^{-1}\wb}{|\wb|^4}A^2_tdt\\
 &-\frac{k_t\wb^\top \ub}{\sigma_{e} v}\sigma^2 \left(\frac{k_t\ub^\top\vb}{\sigma_{e} v}-e_1^\top\vb\right)\frac{\vb^\top \Sigma_t^{-1}\wb}{|\wb|^4}A^2_tdt+\wb^\top\left(e_1-\rb_t\right) dZ_t \ .
\end{align*}
We have that
\begin{align*}
    \wb^\top \ub>0, \quad \frac{k_t\sigma_{e}}{ v}-\vare\gamma \to_{t\to T}\frac{\sigma_{e}}{ v}-\vare\gamma=\Lambda  >0
\end{align*}
and \eqref{eq:sigma-equivalent} shows that 
$$ \frac{\wb^\top \Sigma_t^{-1}\wb}{|\wb|^4}\sim_{t\to T}\frac{1}{\sigma^2 (T-t) \left(\frac{\sigma_{e}^2}{v^2}-\gamma^2\vare^2\right)}\mbox{ and }\vb^\top \Sigma_t^{-1}\wb,\, \vb^\top \Sigma_t^{-1}\vb\mbox{ are bounded.}$$

Thus, the leading term of $dA^2_t$ is $-\frac{k_t\wb^\top \ub}{\sigma_{e} v}\sigma^2 \left(\frac{k_t\sigma_{e}}{ v}-\gamma\vare\right) \frac{\wb^\top \Sigma_t^{-1}\wb}{|\wb|^4}A^2_tdt$ and it is equivalent to $\frac{-c A^2_t}{T-t}$ for some $c>0$ as $t\to T$. Due to this reversion to $0$, $A^2_t\to 0$ as $t\to T$ and 
$\int_0^t k_s\left(\frac{k_s\sigma_{e}}{ v}-\gamma\vare\right) \frac{\wb^\top \Sigma_s^{-1}\wb}{|\wb|^4}A^2_sds$
admits a limit as $t\to T$. Thus, $A^1$ is well defined on $[0,T]$ and admits a limit as $t\to T$. 
Given the orthogonality of $\vb$ and $\wb$, we easily see that $\left(\begin{pmatrix}
    Z_t^\Q-\beta\\ \Xi 
    \end{pmatrix}-\mu_t^\P\right)$ admits a limit as $t\to T$ with 
\begin{align}\label{eq:A2}
    A^2_T:=\lim_{t\to T}A^2_t=\wb^\top\left(\begin{pmatrix}
    Z^\Q_T-\beta\\ \Xi 
    \end{pmatrix}-\mu_T^\P\right)=0 \ .
\end{align}
Thus, the system \eqref{eq:chi}-\eqref{eq:expmup} admits a solution on $[0,T]$.

\medskip

{\textbf{Step 2: Solution to the filtering problem and \eqref{eq:incons}. }}
We first compute $\ab_0$. 
Thanks to the equality  $ k_t' = \gamma \sigma^2 p_t ( \frac{v \varepsilon \gamma}{\sigma_e} -k_t )  $ we have 
\begin{align*}
   &\gamma\sigma^2 q_t \left(k_t- \frac{v\gamma \vare}{\sigma_e}\right)\\
    &=\frac{m_\Xi -\vare\gamma m_\beta - \frac{\sigma_e}{v}m}{\Lambda}\gamma\sigma^2 p_t \left(k_t- \frac{v\gamma \vare}{\sigma_e}\right)=-\frac{m_\Xi -\vare\gamma m_\beta - \frac{\sigma_e}{v}m}{\Lambda}k'_t \ .
\end{align*}
Thus, 
\begin{align*}
    b_t=-m+\int_t^T  \gamma\sigma^2 q_s \left(k_s- \frac{v\gamma \vare}{\sigma_e}\right)ds=-m-\frac{m_\Xi -\vare\gamma m_\beta - \frac{\sigma_e}{v}m}{\Lambda}(k_T-k_t) \ .
\end{align*}
Due to the definition of $m$, this implies that $b_0=0$. Therefore, 
\begin{align}
    \label{eq:ab0}\ab_0=\begin{pmatrix}-m_\beta\\m_\Xi \end{pmatrix} \ .
\end{align}

Let $\Q$ be the probability measure defined in Proposition \ref{p:opt1} with $P_t=P(t,\chi^*_t)=p_t \chi^*_t+q_t$ so that $Z^\Q$ is a Brownian motion under $\Q$. Denote $\mu_t^\Q$ and $\tilde \mu_t^\P$ the $\Fc^Y_t=\Fc^\chi_t$ conditional expectation of $(Z_t^\Q-\beta, \Xi )^\top$ under respectively $\Q$ and $\P$. Our aim is to show that $\mu^\P$ defined in \eqref{eq:expmup} is indeed $\tilde \mu^\P$. 

Similarly to \eqref{eq:eqlaw}, $\Lc^\Q\left( -\beta, \Xi \right) =\Lc^\P\left( -\beta, \Xi  \right)$ so that $\mu_0^\Q=\tilde \mu_0^\P=\begin{pmatrix}-m_\beta\\m_\Xi \end{pmatrix}= \ab_0=\mu_0^\P$. Using
the filtering equations in \cite[Theorem 12.7]{liptser1977statistics} the conditional covariance matrix $\gamma^\P_t$ of $(Z_t^\Q-\beta, \Xi )^\top$ 
satisfies 
$$\frac{d\gamma^\P_t}{dt}=\sigma^2 e_1 e_1^\top-{\sigma^2}\left( e_1+ \gamma^\P_t\Sigma_t^{-1}\left(\rb_t-e_1\right)\right)\left( e_1+ \gamma^\P_t\Sigma_t^{-1}\left(\rb_t-e_1\right)\right)^\top \ . $$
Note that at $t=0$, $\gamma^\P_t\Sigma_t^{-1}$ is the identity matrix. Furthermore, by direct differentiation of the definition of $\Sigma_t$,  
$$\frac{d\Sigma_t}{dt}=\sigma^2 e_1 e_1^\top-{\sigma^2}\left( e_1+ \left(\rb_t-e_1\right)\right)\left( e_1+ \left(\rb_t-e_1\right)\right)^\top $$
and we obtain $\Sigma_t=\gamma^\P_t$.
Applying \cite[Theorem 12.7]{liptser1977statistics} under $\Q$, we easily find that $\Sigma_t$ is also the $\Q$ conditional covariance matrix. 
The same theorem under each probability also leads to 
\begin{align*}
    d\tilde \mu^\P_t&=
    \gamma\sigma^2  P(t,\chi^*_t)
    e_1 dt
    +\rb_t dY_t \ , \\
  d\mu^\Q_t &= \gamma\sigma^2  P(t,\chi^*_t)\rb_t dt-\sigma^2\rb_t \left(\rb_t-e_1\right)^\top \Sigma_t^{-1}(\mu^\Q_t-\mu^\P_t) dt+\rb_t dY_t \ .
\end{align*}
Note that $\tilde \mu^\P_t$ and $ \mu^\P_t$ have the same dynamics and initial conditions. Thus, they are equal. 
The equality $\tilde \mu^\P_t=\mu_t^\P$ and the expression of $X^*$ easily yield the inconspicuousness condition \eqref{eq:incons}.

\medskip

\textbf{Step 3: Optimality of the criterion of the informed trader. }
Taking the difference of $ d \mu^\P_t$ and $ d \mu^\Q_t$ we have
\begin{align}\label{eq:diff}
  d(\mu^\Q_t- \mu^\P_t) =&-\sigma^2\rb_t \left(\rb_t-e_1\right)^\top \Sigma_t^{-1}(\mu^\Q_t-\mu^\P_t) dt-  \gamma\sigma^2 P(t,\chi^*_t)
    \left(e_1-\rb_t \right) dt \ .
\end{align}
Integrating \eqref{eq:diff} and using the equality $\m_0^\P=\mu_0^\Q$ and the definition of $\Phi_{s,t}$ as the flow \eqref{eq:resolvent}, we have 
\begin{align*}
    \mu^\Q_t- \mu^\P_t
    =&- \int_0^t \gamma\sigma^2 P(s,\chi^*_s) \Phi_{s,t}\left(e_1- \rb_s \right) ds \ .
\end{align*}
Thus, 
\begin{align*}
    \mu^\Q_T- \mu^\P_T
    =&- \int_0^T \gamma\sigma^2 P(s,\chi^*_s) \Phi_{s,T}\left(e_1- \rb_T \right) ds- \int_0^T \gamma\sigma^2 P(s,\chi^*_s) \Phi_{s,T}\left(\rb_T- \rb_s \right) ds \ .
\end{align*}
Note that $\rb_T- \rb_s$ is proportional to $\ub$ and $\Phi_{s,T}\ub=\Phi\ub=0.$
By direct computations, we obtain 
$\wb^\top \Phi_{s,T}\left(e_1-\rb_T \right)=\wb^\top \Phi\left(e_1-\rb_T \right)=0$.
Thus, $\wb^\top\mu^\Q_T=\wb^\top\mu^\P_T=\wb^\top\ab_T+\wb^\top\frac{ \ub}{\sigma_{e} v}\chi^*_T=\wb^\top\begin{pmatrix}-m_\beta\\m_\Xi \end{pmatrix}+\frac{\chi^*_T - m }{ \sigma_{e} v}\wb^\top\ub
     \ $, where we used the expression of $\mu_T^\P$ and $\ab_T$. 
We now use \eqref{eq:A2} to obtain 
\begin{align*}
    0&=\wb^\top\left(\mu_T^\P-\begin{pmatrix}Z^\Q_T-\beta\\ \Xi \end{pmatrix}\right)=\wb^\top\left(\mu_T^\Q-\begin{pmatrix}Z^\Q_T-\beta\\ \Xi \end{pmatrix}\right)\\
    &=\frac{\chi^*_T - m }{ \sigma_{e} v}\wb^\top\ub-\wb^\top\begin{pmatrix}Z^\Q_T-\beta+m_\beta\\ \Xi -m_\Xi \end{pmatrix}=\frac{\sigma_{e}}{  v}(\chi^*_T - m )-\wb^\top\begin{pmatrix}Z^\Q_T-\beta+m_\beta\\ \Xi -m_\Xi \end{pmatrix} \ , 
\end{align*}
 and it leads to 
\begin{align*}
    \chi^*_T&= m+ \frac{v}{ \sigma_{e}}\wb^\top\begin{pmatrix}Z^\Q_T-\beta+m_\beta\\ \Xi -m_\Xi \end{pmatrix}=I(Z^\Q_T-\beta, \Xi ;\Nc(m,v^2))
\end{align*}
which is the optimality condition \eqref{eq:target}. Thus, we have the optimality of \eqref{eq:gausscontrol}. 

\medskip

\textbf{Step 4: Martingality of the price and its terminal condition.}
By direct computation 
$$dP(t,\chi^*_t)=\frac{dY_t}{\Lambda^{-1}-{\gamma\sigma^2}(T-t)} \ .$$
Additionally the inconspicuous trading property \eqref{eq:incons} implies that $Y$ is a $\Fc^Y$ martingale. Thus $P(t,\chi^*_t)$ is a $(\P,\Fc^Y)$ martingale with final value 
\begin{align*}
P(T,\chi_T^*)&=p_T \chi^*_T +q_T \\
& =\frac{\sigma_{e}}{v}\chi^*_T -\vare\gamma \chi^*_T -\frac{\sigma_e}{v}m +m_\Xi -\vare \gamma m_\beta \\
&= \wb^\top\begin{pmatrix}Z^\Q_T-\beta+m_\beta\\ \Xi -m_\Xi \end{pmatrix}
-\vare\gamma\chi^*_T
+m_\Xi -\vare \gamma m_\beta \\
&=  \Xi -\gamma \vare (\chi^*_T-Z^\Q_T+\beta)
= \Xi -\gamma \vare (X^*_T+\beta)
\end{align*}
which gives that we indeed have 
$P(t,\chi^*_t)=\E[ \Xi -\gamma \vare (X^*_T+\beta)|\Fc^{X^*+Z}_t].$


\subsection{Proof of Theorem \ref{thm:prop}}
\label{proof-3}

   (1): Eq. \eqref{eq:pricedyn} (and by definition \eqref{eq:kyleslambda}) is a direct consequence of differentiation of $\fP^*_t(  Y )=p_t\chi_t(Y)+q_t.$  The value of the terminal price $\fP^*_T(  Y )$ is obtained in Step 4 of Proof of Theorem \ref{thm:main}.

   (2): The monotonicity and limits of $\Lambda $ as $\vare,\gamma \downarrow 0$ are direct consequences of taking limits in Lemma \ref{lem:v}. 

    (3): For the proof of \eqref{eq:utii}, by Proposition \ref{p:opt1}, when all the random variables are normally distributed, the expected utility of the informed trader conditional to her private information $\{ \Xi = \xi ;\beta=b\}$ is 
\begin{align*}
   & -e^{-\gamma u(0,0)}  
   \E\left[\exp\left(-\gamma \Gamma^c(Z_T-b, \xi )\right)\right]\\
   = &
-e^{-\gamma u(0,0)}  
   \E\Big[\exp\Big(-
   \frac{\gamma}{2} \Big(-2  \xi  (Z_T-b) -{\vare\gamma}(Z_T-b)^2\\
   & \quad \quad \quad
   +\frac{v}{\sigma_{e}}( \xi -m_\Xi  +{\varepsilon \gamma} (Z_T-b+m_\beta))^2 
       +2m ( \xi  +{\varepsilon \gamma} (Z_T-b)) \Big)
   \Big) \Big]  \ . 
\end{align*}
Without loss of generality, let $m_\Xi  = 0$ and $m_\beta = 0$. Then, by completing the square and simplifying, we have
\begin{align*}
   & -e^{-\gamma u(0,0)}  
   \E\left[\exp\left(-\gamma \Gamma^c(Z_T-b, \xi )\right)\right]\\
   = &
-e^{-\gamma u(0,0)}  
   \E\Big[\exp\Big(
\frac{\vare \gamma^2 v \Lambda}{2\sigma_e}\big(
Z_T - b - \frac{\sigma_e}{\vare \gamma v \Lambda} 
(m\vare \gamma -  \xi  \frac{v\Lambda}{\sigma_e})
\big)^2
   \Big) \Big] \\
   & \quad \times \exp\Big(
    -\frac{\gamma}{2}(\frac{ \xi ^2}{\vare \gamma} + \frac{\vare \gamma \sigma_e}{v\Lambda} m^2)
   \Big) \ . 
\end{align*}
Now, since $Z_T \sim N(0, \sigma^2 T)$, we have
\begin{align*}
   &\E\Big[\exp\Big(
\frac{\vare \gamma^2 v \Lambda}{2\sigma_e}\big(
Z_T - b - \frac{\sigma_e}{\vare \gamma v \Lambda} 
(m\vare \gamma -  \xi  \frac{v\Lambda}{\sigma_e})
\big)^2
   \Big) \Big] \\
 = &
 \frac{1}{\sqrt{2\pi \sigma^2 T}}
\int \exp\Big(
\frac{\vare \gamma^2 v \Lambda}{2\sigma_e}\big(
x - b - \frac{\sigma_e}{\vare \gamma v \Lambda} 
(m\vare \gamma -  \xi  \frac{v \Lambda}{\sigma_e})
\big)^2
\Big)
e^{-\frac{x^2}{2 \sigma^2 T}}
dx\\
    = &
\frac{1}{\sqrt{1-\frac{v\Lambda}{\sigma_e}\vare \gamma^2 \sigma^2 T}}    
\exp\left(
\frac{\varepsilon \gamma^2}{2}\frac{ v \Lambda}{\sigma_e} 
\frac{\alpha^2}{1 - \frac{v\Lambda}{\sigma_e}\vare \gamma^2 \sigma^2 T}\right) \ ,
 \end{align*}
 where 
    $$\alpha = 
 b + \frac{\sigma_e}{\varepsilon \gamma v \Lambda} \left( m \varepsilon \gamma -  \xi  \frac{v \Lambda}{\sigma_e} \right) 
 = b + m \frac{\sigma_e}{v\Lambda} - \frac{1}{\vare\gamma} \xi 
 \ . 
    $$
So, we have the expected utility of the informed trader as
\begin{align*}
   -e^{-\gamma u(0,0)} 
   \frac{1}{\sqrt{1-\frac{v\Lambda}{\sigma_e}\vare \gamma^2 \sigma^2 T}}   
   \exp\Big(
    -\frac{\gamma}{2}(\frac{ \xi ^2}{\vare \gamma} + \frac{\vare \gamma \sigma_e}{v\Lambda} m^2)
    + \frac{\varepsilon \gamma^2}{2}\frac{ v \Lambda}{\sigma_e} 
\frac{\alpha^2}{1 - \frac{v\Lambda}{\sigma_e}\vare \gamma^2 \sigma^2 T}
   \Big) \ . 
\end{align*}

    (4): For the proof of \eqref{eq:wealthMM}, we note that $-Y_t$ is the position of the market maker. Thus, the expected profit of the market maker is 
\begin{align}\label{eq:valuemm}
    &\E\left[-Y_T(V_f-P_T)-\int_0^T Y_s dP_s\right]=-\vare\gamma \E\left[Y_T(X_T+\beta)\right]=-\vare\gamma \E\left[Y_T(\chi^*_T-e_1^\top\mu^\P_T)\right]
\end{align}
where we have used the martingality of $P$, the expression for $V_f-P_T$ and the conditioning of $(X_T+\beta)$ in $\Fc^Y_T.$ 
Note that under $\P$, $Y$ is a Brownian motion and 
thanks to \eqref{eq:chi}, 
$$d\chi^*_t= \gamma\sigma^2 (p_t \chi^*_t+q_t)dt+d Y_t \ .$$
Thus, 
\begin{align*}
    d\E\left[Y_t\chi^*_t\right]=\gamma\sigma^2 p_t \E\left[Y_t\chi^*_t\right]dt+\sigma^2 dt \ ,
\end{align*}
which implies that 
$$\E\left[Y_T\chi^*_T\right]=\sigma^2\int_0^Te^{\gamma\sigma^2\int_s^T p_rdr}ds=\sigma^2\int_0^T\frac{1}{\left( 1-\Lambda{\gamma\sigma^2}(T-s)\right)}ds=-\frac{\ln \left( 1-\Lambda{\gamma\sigma^2}T\right)}{\Lambda\gamma} \ . $$
Expanding \eqref{eq:valuemm} using this expression, we obtain that this expected wealth is
\begin{align*}    &-\vare\gamma \E\left[Y_T\left(\chi^*_T(1-\vare\gamma \frac{\sigma^2T+\sigma^2_\beta}{\sigma_{e} v})+m_\beta + \frac{m \gamma \vare (\sigma^2 T + \sigma_\beta^2)}{\sigma_e v}
\right)\right]   \\ 
& =-\vare\gamma \left(1-\vare\gamma \frac{\sigma^2T+\sigma^2_\beta}{\sigma_{e} v}\right)\E\left[Y_T\chi^*_T\right]\\
    &=\frac{\vare}{\Lambda} \left(1-\vare\gamma \frac{\sigma^2T+\sigma^2_\beta}{\sigma_{e} v}\right)\ln(1-\gamma\sigma^2 T \Lambda) <0 \ .
\end{align*}


\bibliographystyle{halpha}
\bibliography{ref}

\newcommand{\etalchar}[1]{$^{#1}$}
\begin{thebibliography}{BHMBO12}

\bibitem[Bac92]{back1992}
Kerry Back.
\newblock Insider trading in continuous time.
\newblock {\em The Review of Financial Studies}, 5(3):387--409, 1992.

\bibitem[Bar02]{baruch}
Shmuel Baruch.
\newblock Insider trading and risk aversion.
\newblock {\em Journal of Financial Markets}, 5(4):451--464, 2002.

\bibitem[BCDF{\etalchar{+}}18]{bal}
Kerry Back, Pierre Collin-Dufresne, Vyacheslav Fos, Tao Li, and Alexander
  Ljungqvist.
\newblock Activism, strategic trading, and liquidity.
\newblock {\em Econometrica}, 86(4):1431--1463, 2018.

\bibitem[BCEL21]{back2020optimal}
Kerry Back, Francois Cocquemas, Ibrahim Ekren, and Abraham Lioui.
\newblock Optimal transport and risk aversion in {K}yle's model of informed
  trading, 2021, 2006.09518.

\bibitem[BD21]{RyanDonnelly}
Weston Barger and Ryan Donnelly.
\newblock Insider trading with temporary price impact.
\newblock {\em International Journal of Theoretical and Applied Finance},
  24(2):2150006, 2021.

\bibitem[BE23]{bose2020kyle}
Shreya Bose and Ibrahim Ekren.
\newblock Kyle-{B}ack models with risk aversion and non-{G}aussian beliefs.
\newblock {\em The Annals of Applied Probability}, 33(6A):4238 -- 4271, 2023.

\bibitem[BE24]{bose2021multidimensional}
Shreya Bose and Ibrahim Ekren.
\newblock Multidimensional {K}yle-{B}ack model with a risk averse informed
  trader.
\newblock {\em SIAM Journal on Financial Mathematics}, 15(1):93--120, 2024.

\bibitem[BHMBO12]{bia}
Francesca Biagini, Yaozhong Hu, Thilo Meyer-Brandis, and Bernt Oksendal.
\newblock Insider trading equilibrium in a market with memory.
\newblock {\em Mathematics and Financial Economics}, 6(3):229--247, 2012.

\bibitem[BP98]{bp}
Kerry Back and Hal Pedersen.
\newblock Long-lived information and intraday patterns.
\newblock {\em Journal of Financial Markets}, 1(3-4):385--402, 1998.

\bibitem[C{\c{C}}07]{cc2007}
Luciano Campi and Umut {\c{C}}etin.
\newblock Insider trading in an equilibrium model with default: a passage from
  reduced-form to structural modelling.
\newblock {\em Finance and Stochastics}, 11(4):591--602, 2007.

\bibitem[C{\c{C}}D11]{ccd}
Luciano Campi, Umut {\c{C}}etin, and Albina Danilova.
\newblock Dynamic {M}arkov bridges motivated by models of insider trading.
\newblock {\em Stochastic Processes and their Applications}, 121(3):534--567,
  2011.

\bibitem[{\c{C}}D16]{cd1}
Umut {\c{C}}etin and Albina Danilova.
\newblock Markovian {N}ash equilibrium in financial markets with asymmetric
  information and related forward--backward systems.
\newblock {\em The Annals of Applied Probability}, 26(4):1996--2029, 2016.

\bibitem[CDF16]{cdf}
Pierre Collin-Dufresne and Vyacheslav Fos.
\newblock Insider trading, stochastic liquidity, and equilibrium prices.
\newblock {\em Econometrica}, 84(4):1441--1475, 2016.

\bibitem[CENV22]{first_paper}
Reda Chhaibi, Ibrahim Ekren, Eunjung Noh, and Lu~Vy.
\newblock A unified approach to informed trading via monge-kantorovich duality,
  2022, 2210.17384.

\bibitem[{\c{C}}et18]{c}
Umut {\c{C}}etin.
\newblock Financial equilibrium with asymmetric information and random horizon.
\newblock {\em Finance and Stochastics}, (1), 2018.

\bibitem[{\c{C}}et23]{ccetin2023insider}
Umut {\c{C}}etin.
\newblock Insider trading with legal risk in continuous time.
\newblock {\em Available at SSRN 4651481}, 2023.

\bibitem[CFDNO10]{corcuera2010kyle}
Jos{\'e}~Manuel Corcuera, Gergely Farkas, Giulia Di~Nunno, and Bernt Oksendal.
\newblock Kyle-{B}ack's model with {L}{\'e}vy noise.
\newblock {\em Preprint series in pure mathematics}, 2010.

\bibitem[Cho03]{cho}
Kyung-Ha Cho.
\newblock Continuous auctions and insider trading: uniqueness and risk
  aversion.
\newblock {\em Finance and Stochastics}, 7(1):47--71, 2003.

\bibitem[CKL23]{ckl23}
Jin~Hyuk Choi, Heeyoung Kwon, and Kasper Larsen.
\newblock Trading constraints in continuous-time {K}yle models.
\newblock {\em SIAM Journal on Control and Optimization}, 61(3):1494--1512,
  2023.

\bibitem[{\c{C}}X13]{cx}
Umut {\c{C}}etin and Hao Xing.
\newblock {Point process bridges and weak convergence of insider trading
  models}.
\newblock {\em Electronic Journal of Probability}, 18:1 -- 24, 2013.

\bibitem[HS94]{hs1994}
Craig~W Holden and Avanidhar Subrahmanyam.
\newblock Risk aversion, imperfect competition, and long-lived information.
\newblock {\em Economics Letters}, 44(1-2):181--190, 1994.

\bibitem[KL24]{cc24}
Christoph Kühn and Christopher Lorenz.
\newblock Insider trading in discrete time {K}yle games.
\newblock {\em Mathematics and Financial Economics}, 2024.

\bibitem[KS12]{karatzas2012brownian}
Ioannis Karatzas and Steven Shreve.
\newblock {\em Brownian motion and stochastic calculus}, volume 113.
\newblock Springer, New York, NY, 2 edition, 2012.

\bibitem[Kyl85]{kyle}
Albert~S Kyle.
\newblock Continuous auctions and insider trading.
\newblock {\em Econometrica: Journal of the Econometric Society}, pages
  1315--1335, 1985.

\bibitem[Las04]{las2007}
Guillaume Lasserre.
\newblock Asymmetric information and imperfect competition in a continuous time
  multivariate security model.
\newblock {\em Finance and Stochastics}, 8(2):285--309, 2004.

\bibitem[LS77]{liptser1977statistics}
Robert~Shevilevich Liptser and Albert~Nikolayevich Shiryaev.
\newblock {\em Statistics of random processes: General theory}, volume 394.
\newblock Springer, New York, NY, 1977.

\bibitem[San15]{santambrogio2015optimal}
Filippo Santambrogio.
\newblock Optimal transport for applied mathematicians.
\newblock {\em Birk{\"a}user, NY}, 55(58-63):94, 2015.

\bibitem[Sub98]{Subrahmanyam98}
Avanidhar Subrahmanyam.
\newblock Transaction taxes and financial market equilibrium.
\newblock {\em The Journal of Business}, 71(1):81--118, 1998.

\bibitem[Vil09]{villani}
C{\'e}dric Villani.
\newblock {\em Optimal transport: old and new}, volume 338.
\newblock Springer Berlin, Heidelberg, 2009.

\bibitem[XS24]{xs24}
Longjie Xu and Yufeng Shi.
\newblock Optimal trading and competition with information in the price impact
  model.
\newblock {\em Quantitative Finance}, 24(6):811--825, 2024.

\end{thebibliography}

\end{document}